\documentclass[11pt]{article}
\usepackage{setspace} 
\usepackage[margin=1in]{geometry}
\usepackage{amsmath,amsthm,amssymb}
\usepackage{algorithm,algorithmic,soul}
\usepackage[fleqn,tbtags]{mathtools}
\usepackage{bbm}
\usepackage{color}
\usepackage{colortbl}
\usepackage{tikz}
\usepackage{subcaption}
\allowdisplaybreaks

\usepackage{mathrsfs}

\usepackage{hyperref}
\usepackage[round]{natbib} 
\usepackage{eurosym}

\newtheorem{lemma}{Lemma}[section]
\newtheorem{proposition}[lemma]{Proposition}
\newtheorem{theorem}[lemma]{Theorem}
\newtheorem{corollary}[lemma]{Corollary}
\newtheorem{definition}[lemma]{Definition}
\newtheorem{example1}[lemma]{Example}
\newtheorem{rem1}[lemma]{Remark}
\newtheorem{assumption}[lemma]{Assumption}
\newtheorem{alg1}[lemma]{Algorithm}
\newtheorem{me1}[lemma]{Mechanism}
\newenvironment{remark}{\begin{rem1}\rm}{\end{rem1}}
\newenvironment{example}{\begin{example1}\rm}{\end{example1}}

\numberwithin{equation}{section}
\newcommand{\R}{\mathbb{R}}

\DeclareMathOperator{\interior}{int}
\DeclareMathOperator{\dom}{dom}

\DeclareMathOperator{\gr}{gr}

\DeclareMathOperator{\esssup}{ess\,sup}
\DeclareMathOperator{\essinf}{ess\,inf}

\let\abs=\envert

\renewcommand{\O}{\Omega}

\newcommand{\F}{\mathcal{F}}

\newcommand{\A}{\mathcal{A}}

\renewcommand{\P}{\mathcal{P}}
\newcommand{\N}{\mathbb{N}}

\renewcommand{\Pr}{\mathbb{P}}

\renewcommand{\a}{\alpha}
\renewcommand{\b}{\beta}

\newcommand{\of}[1]{\ensuremath{\left( #1 \right)}}
\newcommand{\cb}[1]{\ensuremath{ \left\{ #1 \right\} }}

\def\prehp(#1,#2){\ensuremath{  #1 \cdot #2 }}

\newcommand{\E}{\mathbb{E}}

\newcommand{\T}{\top}
\newcommand{\diag}{\operatorname{diag}}
\newcommand{\ind}{\mathbbm{1}}

\DeclareMathOperator*{\argmin}{arg\,min}

\begin{document}

\title{Can Nash inform capital requirements?\\ Allocating systemic risk measures}
\author{\c{C}a\u{g}{\i}n Ararat\thanks{Bilkent University, Department of Industrial Engineering, Ankara, Turkey, cararat@bilkent.edu.tr.}\and  Zachary Feinstein\thanks{Stevens Institute of Technology, School of Business, Hoboken, NJ, USA, zfeinste@stevens.edu.}
}
\date{November 21, 2025}
\maketitle
\abstract{
Systemic risk measures aggregate the risks from multiple financial institutions to find system-wide capital requirements. Though much attention has been given to assessing the level of systemic risk, less has been given to allocating that risk to the constituent institutions. Within this work, we propose a Nash allocation rule that is inspired by game theory. Intuitively, to construct these capital allocations, the banks compete in a game to reduce their own capital requirements while, simultaneously, maintaining system-level acceptability. We provide sufficient conditions for the existence and uniqueness of Nash allocation rules, and apply our results to the prominent structures used for systemic risk measures in the literature. We demonstrate the efficacy of Nash allocations with numerical case studies using the Eisenberg-Noe aggregation mechanism.
}

\section{Introduction}\label{sec:intro}
The financial system is comprised of complex relationships and interconnections between different firms. Because of these connections, \emph{systemic risk} naturally arises, i.e., when one firm or sector is shocked, negative externalities are felt by all other institutions. In its extreme, the distress or failure of one institution can lead to forms of financial contagion which amplify that initial shock throughout the whole of the financial system.
While there are many models of systemic risk, e.g., for default contagion \citep{eisenberg2001systemic,rogers2013failure} or price-mediated contagion~\citep{cifuentes2005liquidity,amini2016uniqueness,feinstein2017financial}, within this work we take a model-agnostic view of systemic risk.

As systemic risk concerns the risk of the entire financial system, rather than \emph{ex ante} focusing on any particular portfolio or firm, many metrics have been developed to assess such risks. 
Taking the idea from monetary risk measures \citep{artzner1999coherent,follmer2002convex}, systemic risk measures have been introduced that provide system-wide capital requirements to guarantee that the outcome to the entire system is acceptable.
Within the literature, there are two broad notions of systemic risk measures:
\begin{itemize}
\item \textbf{Insensitive systemic risk measures:} First introduced in \cite{chen2013axiomatic,kromer2016systemic}, these systemic risk measures take the concept of aggregate-then-allocate. That is, the risks of the financial system are aggregated into a single random variable on which the system-wide capital requirements are determined.
\item \textbf{Sensitive systemic risk measures:} First introduced in \cite{feinstein2017measures,armenti2018multivariate}, these systemic risk measures consider the converse of the insensitive systemic risk measures in that they allocate-then-aggregate. That is, each bank is assessed a capital charge which is included into the financial system before the aggregation to the system-wide risks is determined. Because there may be multiple Pareto-efficient allocations, these systemic risk measures are often presented as set-valued risk measures. Notably, and as highlighted in \cite{feinstein2017measures}, this approach includes seemingly disparate methodologies such as CoVaR \citep{adrian2016covar}. A variation on this approach, in which the capital requirements are placed in a recovery fund to be distributed to any bank during a stress event so as to rescue the system, was introduced in \cite{biagini2019unified}.
\end{itemize}

In this way, systemic risk measures provide either the system-wide capital requirements (i.e., insensitive) or sets of acceptable capital allocations (i.e., sensitive systemic risk measures). However, except in special cases highlighted in the below literature review (Section~\ref{sec:literature}), no general methodology has been proposed to jointly select the capital requirements for all banks. We seek to develop a framework to construct capital requirements for all banks based on any systemic risk measure.

Within this work, we approach this problem with a game-theoretic construction for capital requirements. In utilizing game theory, we naturally capture the competition between banks trying to lower their own capital requirements while still maintaining systemic acceptability. Our goal is to systematically define these Nash allocation rules and determine sufficient conditions so as to guarantee their existence and uniqueness.

In systematizing the construction of systemic risk-aware capital allocations, we wish to highlight three primary contributions of this work.
First, as far as the authors are aware, we are the first to propose a \textbf{generic framework to allocate systemic risk} capital requirements to the banks in the system. As we demonstrate within this work, the approach taken herein applies to all systemic risk measures that have been proposed in the prior literature. Notably, not only is this framework general, but due to its construction via a game, the proposed allocation mechanism is financially meaningful in all settings.
Second, due to the game-theoretic construction, we prove general conditions for \textbf{existence and uniqueness} so that this construction can be guaranteed to function in practice. 
Third, we consider the Eisenberg-Noe aggregation function \cite{eisenberg2001systemic} so as to demonstrate the proposed methodology when applied to financial networks. Within this structure, we provide a \textbf{single convex optimization problem} to compute a Nash allocation. This is in contrast to the traditional methodologies for systemic risk measures on financial networks that suffer from the curse of dimensionality (e.g., \cite{feinstein2017measures}) or utilize approximation methods from, e.g., machine learning \citep{doldi2023multivariate,gonon2024computing,feinstein2023acceptable}. We conclude with numerical case studies to demonstrate the efficacy of the Nash allocations when applied to financial networks.

The remainder of the main part of this work is organized as follows. The introduction is concluded with a literature review of other capital allocation rules in Section~\ref{sec:literature}. Section~\ref{sec:background} provides the general notation and definitions that are used throughout this work. The main results on the construction of a Nash allocation rule, including sufficient conditions for the existence and uniqueness of that rule, are provided in Section~\ref{sec:nash}. Nash allocation rules are further studied in detail and demonstrated in numerical case studies in Section~\ref{sec:EN} on systemic risk measures generated by the Eisenberg-Noe clearing system \citep{eisenberg2001systemic}. Section~\ref{sec:conclusion} concludes the work.

\subsection{Literature Review and Systemic Capital Allocations}\label{sec:literature}

Though the aforementioned systemic risk measures are based on the idea of capital requirements, there has been only limited consideration for how to determine the actual capital requirement to be charged to each institution. The first attempt at determining the capital allocations for sensitive systemic risk measures was proposed in \cite{feinstein2017measures}, and independently developed in \cite{armenti2018multivariate}, via the efficient cash-invariant allocation rules. These allocations are determined by minimizing the total system-wide capital required and they are studied further in, e.g., \cite{fissler2019elicitability}.
In contrast, \cite{brunnermeier2019measuring} presents economic axioms for a capital allocation; they present a construction for such a capital allocation in the setting of a specific sensitive systemic risk measure.

There have been different approaches taken for determining the allocation of the individual capital requirements for the scenario-dependent systemic risk measures. Within \cite{biagini2020fairness}, these allocations are determined via the dual representation of the systemic risk measure where the system health is determined by the sum of bank utilities.
A separate allocation approach is presented in \cite{wang2023optimal} in which multiple optimization problems need to be solved to determine the allocations.

A different approach, in which banks are allowed to directly exchange their risks, is taken in \cite{biagini2021systemic,doldi2022multivariate}. However, in doing so, the authors did not directly consider the question of capital requirement but rather the optimal sharing of an exogenously provided rescue fund.
Though we will directly consider the question of capital requirements, we note that the approach taken herein is most similar to this risk-transfer approach in that we will consider how banks may interact with each other to reduce their capital requirements.

Finally, due to the use of risk measures to compute systemic capital requirements, it may be tempting to apply a \emph{vector-valued} risk measure to provide individualized capital requirements while maintaining the interconnections between firms. However, as proven in~\cite{ararat2024separability}, no lower semicontinuous, convex vector-valued risk measure can admit capital requirements that include systemic interactions. That is, following a vector-valued risk measure, the capital requirements for every bank can only depend on its own financial situation. In this way, such a methodology is incapable of capturing the interconnections that are inherent in financial systems.

\section{Background and Notation}\label{sec:background}

Before discussing the relevant background on risk measures, we provide some overarching notation that is used throughout this work. Let us consider the Euclidean space $\R^N$, where $N\in\N:=\{1,2,\ldots\}$. For each $i\in[N]\coloneqq\{1,\ldots,N\}$, we denote by $e_i$ the $i^{\text{th}}$ standard unit vector in $\R^N$. Given $x=(x_1,\ldots,x_N)^{\T}\in \R^N$, $m\in\R$, $i\in[N]$, we define $(m,x_{-i})\in\R^N$ to be the vector whose $i^{\text{th}}$ component is $m$ and who is equal to $x$ in every other component. To compare two vectors $x,y \in \R^N$ componentwise, we write $x\leq y$ if $x_i \leq y_i$ for each $i\in[N]$. In this case, we also define the hyperrectangle $[x,y]:=\{z\in\R^N\; | \; x\leq z\leq y\}$. We say that a function $f\colon \R^N\to \bar\R\coloneqq [-\infty,+\infty]$ is \emph{increasing} if $x\leq y$ implies $f(x)\leq f(y)$ for every $x,y\in\R^N$; we say that $f$ is \emph{decreasing} if $-f$ is increasing. Since we will deal with concave aggregation functions in the sequel, we define the \emph{effective domain} of a function $f\colon \R^N\to \bar\R$ as 
\[
\dom f\coloneqq \{x\in\R^N\; | \; f(x)>-\infty\}.
\]
For a sequence $(x^n)_{n\in\N}$ and a number $x$ in $\bar{\R}$, we write $x^n \uparrow x$ ($x^n \downarrow x$) if $x^1\leq x^2\leq \ldots$ ($x^1\geq x^2\geq \ldots$) and $\lim_{n\rightarrow\infty}x^n=x$. We denote by $I_N$ the $N\times N$ identity matrix. For a matrix $a\in\R^{N\times M}$, where $M\in\N$, we denote by $a_{i\cdot}$ its $i^{\text{th}}$ row (treated as a column vector) and by $a_{\cdot j}$ its $j^{\text{th}}$ column for each $i\in[N]$ and $j\in[M]$. For sets $A,B\subseteq\R^N$, we define their \emph{Minkowski sum} as $A+B:=\{x+y\; | \; x\in A, y\in B\}$ with the convention $A+\emptyset=\emptyset+B=\emptyset$. Then, the collection of all \emph{upper sets} in $\R^N$ is given by
\begin{equation}\label{eq:upperset}
\P(\R^N;\R^N_+) \coloneqq \{A \subseteq \R^N \; | \; A + \R^N_+ = A\},
\end{equation}
where $\R^N_+:=\{x\in\R^N\; | \; 0\leq x\}$ is the cone of positive vectors. We also define the collection of all \emph{rectangular upper sets} in $\R^N$ by
\begin{equation}\label{eq:rectset}
\mathcal{R}(\R^N;\R^N_+)\coloneqq \cb{\bigtimes_{i=1}^N A_i\mid \forall i\in[N]\colon A_i=\R \text{ or }A_i=[x_i,+\infty) \text{ for some }x_i\in\R}\cup\{\emptyset\}.
\end{equation}
Clearly, $\mathcal{R}(\R^N;\R^N_+)\subseteq \P(\R^N;\R^N_+)$.

Let $(\O,\F,\Pr)$ be a probability space. We denote the space of all (equivalence classes of) $\F$-measurable $N$-dimensional random vectors $X = (X_1,\ldots,X_N)^{\T}$ by $L^0(\R^N)$. In particular, (in)equalities and limits about random variables are understood in the $\mathbb{P}$-almost sure (a.s.) sense unless stated otherwise. For a random variable $Y\in L^0(\R)$, we define its $\Pr$-essential supremum as $\esssup Y:=\inf\{c>0\; | \; \Pr\{Y\leq c\}=1\}$ and $\Pr$-essential infimum as $\essinf Y:=-\esssup (-Y)$. For mathematical convenience, we often restrict ourselves to the space $L^\infty(\R^N) \subseteq L^0(\R^N)$ of essentially bounded random vectors, i.e., the set of all $X \in L^0(\R^N)$ such that $\|X_i\|_\infty := \max\{\esssup X_i, -\essinf X_i\} < +\infty$ for every $i \in [N]$. With abuse of notation, we often write $\esssup X:=(\esssup X_1,\ldots,\esssup X_N)^{\T}\in\R^N$ and $\essinf X:=(\essinf X_1,\ldots,\essinf X_N)^{\T}\in\R^N$ %$\|X\|_\infty := (\|X_1\|_\infty,\ldots,\|X_N\|_\infty)^\T$ to be the vector of norms 
for $X \in L^\infty(\R^N)$.

\subsection{Monetary Risk Measures}\label{sec:background-riskmsr}

Let us recall the definition of a risk measure in the univariate setting.

\begin{definition}\label{defn:riskmsr}
A functional $\rho\colon L^\infty(\R)\to\R$ is called a \textbf{risk measure} if it satisfies the following properties:
\begin{itemize}
\item \textbf{Normalized:} $\rho(0) = 0$.
\item \textbf{Monotone:} $Y^1 \geq Y^2$ implies $\rho(Y^1) \leq \rho(Y^2)$ for every $Y^1,Y^2\in L^\infty(\R)$.
\item \textbf{Translative:} $\rho(Y+m) = \rho(Y)-m$ for every $Y\in L^\infty(\R)$ and $m\in\R$.
\end{itemize}
In addition, we usually consider the following properties for $\rho$:
\begin{itemize}
\item \textbf{Convex:} $\rho(\alpha Y^1+(1-\alpha)Y^2)\leq \alpha \rho(Y^1)+(1-\alpha)\rho(Y^2)$ for every $Y^1,Y^2\in L^\infty(\R)$ and $\alpha\in[0,1]$.
\item \textbf{Positively homogeneous:} $\rho(\alpha Y)=\alpha \rho(Y)$ for every $Y\in L^\infty(\R)$ and $\alpha\geq 0$.
\item \textbf{Continuous from above:} $Y^n \downarrow Y$ implies $\rho(Y^n)\uparrow \rho(Y)$ for every $Y,Y^1,Y^2,\ldots\in L^\infty(\R)$.
\end{itemize}
A risk measure $\rho$ is called \textbf{coherent} if it is convex and positively homogeneous.
\end{definition}

Given a risk measure $\rho$, the set
\[
\A := \{Y \in L^\infty(\R) \; | \; \rho(Y) \leq 0\}
\]
is called its \emph{acceptance set}. As recalled in the next proposition, this set characterizes $\rho$.

\begin{proposition}\label{prop:riskmsr-primal}
\citep[Proposition~4.6]{fs:sf} Let $\rho\colon L^\infty(\R)\to\R$ be a risk measure with acceptance set $\A$. Then, for every $Y\in L^\infty(\R)$, it holds
\[
\rho(Y) = \inf\{m \in \R \; | \; Y + m \in \A\}.
\]
Moreover, $\rho$ is convex (positively homogeneous, coherent) if and only if $\A$ is a convex set (cone, convex cone); whenever $\rho$ is convex, $\rho$ is continuous from above if and only if $\A$ is closed in the weak* topology on $L^\infty(\R)$.
\end{proposition}

\subsection{Systemic Risk Measures}\label{sec:background-systemic}

Consider a financial system with $N\in\N$ institutions. In this multivariate setting, we will consider systemic risk measures that are defined as set-valued functionals on the space $L^\infty(\R^N)$. Here, each $X\in L^\infty(\R^N)$ is called a \emph{random shock} or a \emph{stress scenario}, e.g., for each $i\in [N]$, the random variable $X_i$ represents the future worth of the assets of institution $i$, which affects its ability to meet its contractual obligations.

We start by introducing the notion of aggregation function in our setting. For a function $\Lambda\colon\R^N\times\R^N\to\R\cup\{-\infty\}$ and $x,m\in\R^N$, we denote by $\Lambda|_{m}$ the section $x\mapsto \Lambda(x,m)$ and by $\Lambda|^{x}$ the section $m\mapsto \Lambda(x,m)$. Recall that we consider the effective domain $\dom\Lambda = \{(x,m) \in \R^N \times \R^N \; | \; \Lambda(x,m) \in\R\}$; in particular, $\Lambda(x,m) = -\infty$ for $(x,m) \not\in \dom\Lambda$.

\begin{definition}\label{defn:agg}
A function $\Lambda\colon\R^N\times\R^N\to\R\cup\{-\infty\}$ is called an \textbf{aggregation function} if it satisfies the following properties:
\begin{enumerate}
\item \textbf{Monotone:} $\Lambda$ is increasing.
\item \textbf{Concave:} $\Lambda$ is a concave function with a nonempty and closed domain $\dom\Lambda$.
\item \textbf{Continuous:} $\Lambda$ is continuous relative to $\dom\Lambda$.
\item \textbf{Rectangular domain:} (i) For each $m\in\R^N$, it holds $\dom\Lambda|_{m}\in \mathcal{R}(\R^N;\R^N_+)$.\\
(ii) For each $x\in\R^N$, it holds $\dom\Lambda|^{x}\in\mathcal{R}(\R^N;\R^N_+)$.
\end{enumerate}
A function $\bar\Lambda\colon\R^N\to\R\cup\{-\infty\}$ is called a \textbf{single-element aggregation function} if $(x,m)\mapsto \bar{\Lambda}(x)$ is an aggregation function.
\end{definition}

\begin{remark}\label{rem:single}
In view of Definition~\ref{defn:agg}, a function $\bar{\Lambda}\colon\R^N\to\R\cup\{-\infty\}$ is a single-element aggregation function if and only if it is increasing, concave, and $\dom\bar\Lambda$ is a nonempty rectangular upper set. In particular, part (ii) of the rectangular domain property becomes redundant here.
\end{remark}

\begin{remark}
For an aggregation function $\Lambda$, the quantity $\Lambda(x,m)$ can be seen as the impact of a shock of magnitude $x\in \R^N$ to society under capital level $m\in\R^N$. More generally, we can easily extend Definition~\ref{defn:agg} to functions $\Lambda\colon \R^N \times \R^M \to \R \cup \{-\infty\}$, where $M\in\N$. In this way, the shocks need not be of the same dimension as the capital requirements. This applies if, e.g., we consider factor models to generate the shocks (see \cite{banerjee2022pricing}) or groupings of the institutions for capital requirements (see \cite{feinstein2017measures,meimanjan2023}).
\end{remark}

We will discuss single-element aggregation functions later starting at Example~\ref{ex:syst}. For our general framework, let $\A \subseteq L^\infty(\R)$ be the acceptance set of a risk measure $\rho$ and $\Lambda\colon\R^N \times \R^N \to \R \cup \{-\infty\}$ be an aggregation function. For each $X\in L^\infty(\R^N)$, we define the set
\[
D(X)\coloneqq \{m\in\R^N\; | \; \mathbb{P}\{(X,m)\in \dom\Lambda\}=1\},
\]
i.e., $\Lambda(X,m)$ is a real-valued random variable for each $m\in D(X)$. In the sequel, it will be convenient to introduce the set
\[
L^\infty_\Lambda(\R^N)\coloneqq \cb{X\in L^\infty(\R^N)\; | \; D(X)\neq\emptyset}.
\]
Since $\dom\Lambda\neq\emptyset$, it is easy to verify that $L^\infty_\Lambda(\R^N)\neq\emptyset$ as well.

\begin{lemma}\label{asmp:domain}
(i) For each $X\in L^\infty(\R^N)$, we have $D(X)=D(\essinf X) = \dom\Lambda|^{\essinf X}$.\\
(ii) For each $X\in L_\Lambda^\infty(\R^N)$ and $m\in D(X)$, it holds $\Lambda(X,m)\in L^\infty(\R)$.
\end{lemma}

\begin{proof}
(i) Let $X\in L^\infty(\R^N)$. Since $\essinf X\leq X$ a.s., we have $D(\essinf X)\subseteq D(X)$. The converse is trivial if $D(X)=\emptyset$. Otherwise, let $m\in D(X)$. Then, $\Pr\{X\in \dom\Lambda|_{m}\}=1$. In particular, $\dom\Lambda|_{m}$ is a nonempty set. Since $\dom\Lambda|_{m}\in \mathcal{R}(\R^N_+;\R^N_+)$, there exists $\underline{x}\in\dom\Lambda|_{m}$ such that $\underline{x}\leq X$ a.s. Then, $\underline{x}\leq \essinf X$ so that $-\infty<\Lambda(\underline{x},m)\leq \Lambda(\essinf X, m)$ by the monotonicity of $\Lambda$. Hence, $m\in D(\essinf X)$ and we get $D(X)\subseteq D(\essinf X)$. \\
(ii) Let $X\in L_\Lambda^\infty(\R^N)$ and $m\in D(X)$. By (i), we have $m\in D(\essinf X)$ so that $-\infty<\Lambda(\essinf X,m)\leq \Lambda(X,m)\leq \Lambda(\esssup X,m)$ a.s. by monotonicity. Hence, we obtain $\Lambda(X,m)\in L^\infty(\R)$.
\end{proof}

\begin{definition}\label{defn:systrisk}
The set-valued functional $R\colon L^\infty_\Lambda(\R^N)\rightrightarrows \R^N$ defined by
\[
R(X) \coloneqq \{m \in D(X) \; | \; \Lambda(X,m) \in \A\},
\]
for each $X\in L_\Lambda^\infty(\R^N)$, is called the \textbf{$(\A,\Lambda)$-systemic risk measure}.
\end{definition}

In the above definition, if a stress scenario $X \in L_\Lambda^\infty(\R^N)$ is given, then $R(X)$ provides the set of all acceptable \emph{systemic} capital allocations, i.e., accounting for interactions between institutions as modeled by the aggregation function $\Lambda$.

\begin{remark}\label{rem:upperset}
Since $m\mapsto \Lambda(x,m)$ is increasing with a rectangular domain for each $x\in\R^N$, the set-valued function $D$ and the $(\A,\Lambda)$-systemic risk measure $R$ map into the collection of rectangular upper and upper sets, respectively, i.e., $D(X) \in \mathcal{R}(\R^N;\R^N_+)$ and $R(X)\in \P(\R^N;\R^N_+)$ for each $X\in L_\Lambda^\infty(\R^N)$; see \eqref{eq:upperset} and \eqref{eq:rectset}. Moreover, since $x\mapsto \Lambda(x,m)$ is increasing for each $m\in\R^N$ as well, $D$ and $R$ have the following monotonicity property: $D(X^1)\subseteq D(X^2)$ and $R(X^1)\subseteq R(X^2)$ for every $X^1,X^2\in L^\infty_\Lambda(\R^N)$ with $X^1\leq X^2$.
\end{remark}

We proceed with examples of systemic risk measures that have been studied in the recent literature.

\begin{example}\label{ex:syst}
Within the literature, systemic risk measures tend to follow the forms discussed below. Both of these constructions are derived from a single-element aggregation function $\bar\Lambda\colon \R^N \to \R \cup \{-\infty\}$.
\begin{itemize}
\item \textbf{Insensitive systemic risk measures:} Let us take $\Lambda=\Lambda^I$, where
\begin{equation}\label{eq:lambda-ins}
\Lambda^I(x,m) := \bar\Lambda(x) + \sum_{i \in [N]} m_i
\end{equation}
for each $x,m\in\R^N$. Clearly, $\Lambda^I$ is an aggregation function and $\dom\Lambda^I=\dom\bar\Lambda\times\R^N$. Let
\[
L^\infty(\dom \bar\Lambda)\coloneqq \{X\in L^\infty(\R^N)\; | \; \mathbb{P}\{X\in\dom\bar\Lambda\}=1\}.
\]
Then, we have $D^I(X):=D(X)=\emptyset$ if $X\notin L^\infty(\dom \bar\Lambda)$ and $D^I(X):=D(X)=\R^N$ if $X\in L^\infty(\dom \bar\Lambda)$. In particular, $L^\infty_{\Lambda^I}(\R^N)=L^\infty(\dom\bar\Lambda)$. For every $X\in L^\infty(\dom\bar\Lambda)$, by Lemma~\ref{asmp:domain}, we have $\bar\Lambda(X)\in L^\infty(\R)$ and
\[
R^I(X)\coloneqq R(X)\negthinspace=\negthinspace\cb{m\in\R^N\; | \; \bar\Lambda(X)+\sum_{i\in[N]} m_i \in \A}\negthinspace=\negthinspace \cb{m\in\R^N\; | \; \rho(\bar\Lambda(X))\leq \sum_{i\in[N]} m_i};
\]
in particular, $R^I(X)$ is a halfspace. Hence, the set-valued functional $R^I$ essentially reduces to the real-valued functional $X\mapsto \rho(\bar\Lambda(X))$ on $L^\infty(\dom \bar\Lambda)$ as studied in \cite{chen2013axiomatic,kromer2016systemic}.
\item \textbf{Sensitive systemic risk measures:} Let us take $\Lambda=\Lambda^S$, where
\begin{equation}\label{eq:lambda-sen}
\Lambda^S(x,m) \coloneqq \bar\Lambda(x+m)
\end{equation}
for each $x,m\in\R^N$. Clearly, $\Lambda^S$ is an aggregation function with
\begin{equation}\label{eq:domsen}
\dom\Lambda^S=\{(x,m)\in\R^N\times\R^N\; | \; x+m\in\dom\bar\Lambda\}.
\end{equation}
Then, for each $X\in L^\infty(\R^N)$, we have
\begin{equation}\label{eq:Dsen}
D^S(X):=D(X)=\{m\in\R^N\; | \; \mathbb{P}\{X+m\in \dom \bar\Lambda\}=1\}.
\end{equation}
Since $\dom\bar\Lambda\neq\emptyset$ and $\bar\Lambda$ is increasing, it follows that $D^S(X)\neq \emptyset$ for every $X\in L^\infty(\R^N)$. In particular, $L^\infty_{\Lambda^S}(\R^N)=L^\infty(\R^N)$. By Lemma~\ref{asmp:domain}, we have $\bar\Lambda(X+m)\in L^\infty(\R)$ for every $X\in L^\infty(\R^N)$ and $m\in D^S(X)$. Then,
\begin{align*}
R^S(X)\coloneqq R(X)&=\{m\in D^S(X)\; | \; \bar\Lambda(X+m)\in \A\}\\
&=\{m\in \R^N\; | \; \rho(\bar\Lambda(X+m))\leq 0,\ \essinf X+m\in \dom\bar\Lambda\}
\end{align*}
for every $X\in L^\infty(\R^N)$, which yields the systemic risk measures studied in \cite{feinstein2017measures,ararat2020dual}. Compared to the previous case, these systemic risk measures consider the direct impact of capital allocations on the stress scenario before aggregation, hence they are more sensitive to capital levels.
\end{itemize}
\end{example}

\begin{remark}\label{rem:sd}
A variation of the sensitive aggregation function has been studied in \cite{biagini2019unified,biagini2020fairness,biagini2021systemic} which allow for the distribution of capital to be be a random variable, rather than deterministic. Herein, we refer to this special type of aggregation function \textbf{scenario-dependent}. To be precise, let $\bar\Lambda\colon \R^N \to \R \cup \{-\infty\}$ be a single-element aggregation function with $\dom\bar\Lambda=\underline{x}+\R^N_+$ for some $\underline{x}\in\R^N$. We define
\begin{equation}\label{eq:lambda-coop}
\bar\Lambda^{SD}(x) := \begin{cases}\sup_{z \in \R^N}\{\bar\Lambda(z) \; | \; \sum_{i \in [N]} z_i \leq \sum_{i \in [N]} x_i\}& \text{if }x\in \dom\bar\Lambda,\\ -\infty & \text{if } x\in\R^N\setminus\dom\bar\Lambda.\end{cases}
\end{equation}
Note that, when $x\in \dom\bar\Lambda$, the feasible region of the maximization problem in \eqref{eq:lambda-coop} is a nonempty compact set, a maximizer always exists, and we have $\bar\Lambda^{SD}(x)\in\R$. In particular, $\dom\bar\Lambda^{SD}=\dom\bar\Lambda$. Since $\bar\Lambda$ is continuous relative to $\dom\bar{\Lambda}$, it follows from \citet[Theorem~1.17(c)]{rockafellarwets} that $\bar{\Lambda}^{SD}$ is continuous relative to $\dom\bar{\Lambda}^{SD}$. Since $\bar{\Lambda}$ is concave, the concavity of $\bar\Lambda^{SD}$ follows directly by \citet[Proposition~2.22(a)]{rockafellarwets}. Hence, $\bar\Lambda^{SD}$ is a single-element aggregation function. It follows that the sensitive systemic risk measure corresponding to $\bar\Lambda^{SD}$ has the form
\begin{align*}
R^{SD}(X)&\coloneqq \{m\in D(X)\;|\;\bar\Lambda^{SD}(X+m)\in\A\}\\
&=\cb{m \in D(X)\; | \; \exists Y \in L^\infty(\R^N): \, \bar\Lambda(X+Y) \in \A, \, \sum_{i \in [N]} Y_i \leq \sum_{i \in [N]} m_i},
\end{align*}
where $D(X):=\{m\in\R^N\;|\;\mathbb{P}\{X+m\in\dom\bar\Lambda\}=1\}$ for every $X \in L^\infty(\R^N)$. In particular, $R^{SD}(X)$ is the intersection of a halfspace with $D(X)$. Hence, the set-valued functional $R^{SD}$ essentially reduces to the minimum sum scalarization $X \mapsto \inf\{\sum_{i \in [N]} m_i \; | \; m \in R^{SD}(X)\}$ as studied in \cite{biagini2019unified,biagini2020fairness,biagini2021systemic}.
\end{remark}

We conclude our review of systemic risk measures by providing examples of single-element aggregation functions that we will revisit in Section~\ref{sec:nash} below. We note that we also consider a different aggregation function based on financial clearing conditions in Section~\ref{sec:EN}.

\begin{example}\label{ex:se-aggregation}
Here are a few examples of single-element aggregation functions:
\begin{enumerate}
\item\label{util} \textbf{Sum of univariate utilities:} For each $i\in[N]$, let $u_i\colon\R\to\R\cup\{-\infty\}$ be a \emph{univariate utility function}, i.e., an increasing, concave, and upper semicontinuous function such that $\dom u_i\neq \emptyset$ is closed. Let $\bar{\Lambda}^U(x)\coloneqq\sum_{i \in [N]} u_i(x_i)$ for each $x\in\R^N$. As presented in, e.g., \cite{biagini2019unified,biagini2020fairness}, the function $\bar\Lambda^U$ is a single-element aggregation function with $\dom\bar\Lambda^U=\bigtimes_{i \in [N]} \dom u_i\neq\emptyset$.
\item\label{mf} \textbf{Mean-field utility:} Inspired by mean-field games (see, e.g., \citet[Section~3.1]{carmonaMF} and \citet[Section~2]{bensoussanLQ}), we introduce the idea that the system's utility may depend both on each bank's individual utility as well as the average health of the financial system. Mathematically, consider the univariate utility functions $u_i: \R \to \R \cup \{-\infty\}$ for each bank $i \in [N]$ and a systemic externality utility $\bar u: \R \to \R \cup \{-\infty\}$. 
In this way, we augment the sum of univariate utilities by this systemic externality:
\[
\bar\Lambda^{MF}(x) := \sum_{i \in [N]} u_i(x_i) + \bar u\left(\frac{1}{N}\sum_{i \in [N]} x_i\right)
\]
for each $x \in \R^N$. 
Assuming that $\bigtimes_{i \in [N]} \dom u_i \subseteq \{x \in \R^N \; | \; \frac{1}{N}\sum_{i \in [N]} x_i \in \dom\bar u\}$, this single-element aggregation function has domain $\dom\bar\Lambda^{MF} = \bigtimes_{i \in [N]} \dom u_i$.
\end{enumerate}
\end{example}

\section{Nash Allocation Rules}\label{sec:nash}

In this section, we wish to construct a financially meaningful capital allocation rule that accounts for the contribution that each firm makes towards the risk of the full system. We are specifically motivated by the aggregation functions that are sums of univariate utilities (see Example~\ref{ex:se-aggregation}\eqref{util}), i.e.,
$\bar\Lambda(x) = \sum_{i \in [N]} u_i(x_i)$, $x\in\R^N$, for univariate utility functions $u_1,\ldots,u_N\colon \R \to \R \cup \{-\infty\}$ so that each bank's contribution to the system's utility is clear. We aim to generalize this concept for broader classes of aggregation functions. Notably, in contrast to the \emph{set-valued} systemic risk measures introduced in Section~\ref{sec:background-systemic}, which provide the collection of all acceptable capital requirements, we focus on finding specific, i.e., \emph{vector-valued}, acceptable capital allocations instead.

Throughout this section, let $\A$ be the acceptance set of a risk measure $\rho$ and $\Lambda\colon \R^N\times\R^N\to\R\cup\{-\infty\}$ be an aggregation function. We consider the corresponding $(\A,\Lambda)$-systemic risk measure $R$ given by Definition~\ref{defn:systrisk}. We will work under the following assumption for $\A$ and $\rho$:

\begin{assumption}\label{asmp:coherent}
The risk measure $\rho$ is a coherent risk measure that is continuous from above; equivalently, the acceptance set $\A$ is a convex cone that is closed with respect to the weak* topology on $L^\infty(\R)$.
\end{assumption}

As a generalization of aggregation functions that are in the form of a sum of univariate utility functions, we assume that $\Lambda$ can be expressed in the form
\begin{equation}\label{eq:decomposition}
\Lambda(x,m)=\sum_{i\in[N]} \Lambda_i(x,m)
\end{equation}
for each $(x,m)\in \R^N\times\R^N$, where $\Lambda_1,\ldots,\Lambda_N\colon\R^N\times\R^N\to\R\cup\{-\infty\}$ are aggregation functions that satisfy the following assumption:

\begin{assumption}\label{asmp:decomposition}
\begin{enumerate}
\item \textbf{Consistent domain:} For each $i\in[N]$, it holds $\dom\Lambda_i=\dom\Lambda$.
\item \textbf{Self-feasible:} For every $x\in \R^N$ with $D(x) = \dom\Lambda|^{x} \neq \emptyset$, there exists $\bar{r} \in D(x)$ such that $\Lambda_i(x,(\bar{r}_i,m_{-i})) > 0$ for every $m \in D(x)$ and $i \in [N]$.
\item \textbf{Self-infeasible:} For every $x\in \R^N$ with $D(x) = \dom\Lambda|^{x} \neq \emptyset$, there exists $\underline{r}\in D(x)$ such that $-\infty\leq \Lambda_i(x,(\underline{r}_i-\varepsilon,m_{-i}))<0$ for every $m\in D(x)$, $i\in[N]$, and $\varepsilon>0$.
\end{enumerate}
\end{assumption}

In what follows, we refer to $(\Lambda_i)_{i\in[N]}$ as a \textbf{decomposition} of $\Lambda$. We say that $(\Lambda_i)_{i\in[N]}$ is \textbf{self-preferential} if there exists some constant $L \in [0,1)$ such that, for every $i,j\in[N]$ with $i\neq j$, $(x,m)\in\dom\Lambda$, and $\delta\geq 0$, it holds 
\[
\Lambda_i(x,m+L\delta e_i) \geq \Lambda_i(x,m+\delta e_j).
\]

\begin{remark}
The self-preferential property of a decomposition means that increasing the capital to the bank associated with the aggregation function is more efficient (by a fixed factor) than adding it to another bank. In particular, from the perspective of a Nash game, no bank has an incentive to bail out any other bank; instead, banks will retain any capital they are allocated. This idea will be formally introduced in Definition~\ref{defn:nash-allocation} below.
\end{remark}

\begin{remark}\label{rem:select}
Given the aggregation function $\Lambda$, choosing a suitable decomposition for it in the form of \eqref{eq:decomposition} may be a nontrivial task. 
For instance, the decompositions defined by the Aumann-Shapley value (see \cite{denault}) or the ``with-minus-without" allocations considered in \cite{brunnermeier2019measuring} do \emph{not} satisfy the properties of an aggregation function in general. This is despite these constructions satisfying other useful properties in general.

However, in certain cases, each bank $i\in[N]$ can have a natural choice of an aggregation function, say $\lambda_i$, with a clear economic interpretation, e.g., the utility or net repayment realized by bank $i$. If such a natural decomposition is not available a priori, then we may simply take $\lambda_i\equiv 0$. In general, the mismatch between $\Lambda$ and the sum of individual aggregation functions, i.e., the function $\Delta:=\Lambda-\sum_{i\in[N]}\lambda_i$ is yet to be decomposed as $\Delta=\sum_{i\in[N]}\delta_i$ for suitably selected $\delta_1,\ldots,\delta_N$ so that we can set $\Lambda_i:=\lambda_i+\delta_i$ for each $i\in[N]$ to establish the structure in \eqref{eq:decomposition}. In Appendix~\ref{app:nash-decomposition}, we will address the decomposition selection problem for a generic aggregation function $\Delta\colon\R^N\times\R^N\to\R\cup\{-\infty\}$.
\end{remark}

Next, we illustrate the use of decompositions through the motivating example of sum of univariate utilities aggregation. While a general theory for defining meaningful decompositions for generic aggregation functions will be provided in Appendix~\ref{app:nash-decomposition}, decompositions for single-element aggregation functions are defined and studied systematically in Section~\ref{sec:single-element}.

\begin{example}\label{ex:sen-ut}
Recall the single-element aggregation function $\bar\Lambda^U$ in Example~\ref{ex:se-aggregation}\eqref{util} that is in the form of sum of univariate utility functions $u_1,\ldots,u_N$. For this example, consider the corresponding sensitive aggregation function $\Lambda^{S}$ defined by
\[
\Lambda^{S}(x,m)\coloneqq \bar\Lambda^U(x+m)=\sum_{i\in[N]}u_i(x_i+m_i)
\]
for each $x,m\in\R^N$. For each $i \in [N]$, let us define
\[
\Lambda^{S}_i(x,m)\coloneqq u_i(x_i+m_i)
\]
if $x,m\in\R^N$ satisfy $x+m\in \dom\bar\Lambda^U$ and $\Lambda^{S}_i(x,m)\coloneqq -\infty$ otherwise. Clearly, $\Lambda^{S}_1,\ldots,\Lambda^{S}_N$ are aggregation functions, $(\Lambda^{S}_i)_{i\in[N]}$ satisfies the consistent domain property, and $\sum_{i\in[N]}\Lambda^{S}_i = \Lambda^{S}$. Let us further assume that $\{x_i \in \dom u_i \; | \; u_i(x_i) \geq 0\} \not\in \{\emptyset,\R\}$ for each $i\in[N]$, which guarantees self-feasibility and self-infeasibility for $(\Lambda^{S}_i)_{i\in[N]}$. Hence, $(\Lambda^{S}_i)_{i\in[N]}$ is a decomposition of $\Lambda^{S}$. Moreover, this decomposition is self-preferential with $L=0$ due to the independence of the utility functions.
\end{example}

Using a decomposition of $\Lambda$, we are ready to define a Nash allocation rule, the central object of the paper. 

\begin{definition}\label{defn:nash-allocation}
Let $(\Lambda_i)_{i\in[N]}$ be a decomposition of $\Lambda$. A vector-valued functional $r\colon L_\Lambda^\infty(\R^N) \to \R^{N}$ is called a $(\A,(\Lambda_i)_{i\in[N]})$\textbf{-Nash allocation rule} if, for every $X\in L_\Lambda^\infty(\R^N)$, it holds
\begin{equation}\label{eq:NE}
r_i(X) = \inf\{m_i \in \R \; | \; \Lambda_i(X,(m_i,r_{-i}(X))) \in \A\}
\end{equation}
for each $i\in[N]$.\footnote{In the infimum, we implicitly consider $m_i\in\R$ such that $(m_i,r_{-i}(X))\in D(X)$, which ensures that $\Lambda_i(X,(m_i,r_{-i}(X)))\in L^\infty(\R)$ by Lemma~\ref{asmp:domain}.} In this case, $r$ is called \textbf{acceptable} if $r(X)\in R(X)$ for every $X\in L_\Lambda^\infty(\R^N)$.
\end{definition}

In Appendix~\ref{sec:rescue}, we propose a generalization of a Nash allocation rule that permits banks to use their capital to directly bail out or rescue other institutions. However, for all examples considered herein, all Nash allocations with rescue funds satisfy the above definition (see Proposition~\ref{prop:self-preferential}); hence, we focus on this simpler setting in this paper. We leave a more formal study of this generalization for future work.

Within this section, we consider general conditions for the acceptability of a Nash allocation rule (Section~\ref{sec:nash-acceptable}). With conditions so that acceptability is guaranteed, we provide a general theorem for the existence and uniqueness of Nash allocation rules (Section~\ref{sec:nash-exist}). With these results for general systemic risk measures, we provide further insights for those systemic risk measures generated from a single-element aggregation function, i.e., insensitive (Section~\ref{sec:nash-insensitive}) and sensitive (Section~\ref{sec:nash-sensitive}) systemic risk measures.

\subsection{Acceptable Allocation Rules}\label{sec:nash-acceptable}

Before focusing on Nash allocation rules specifically, we want to consider the acceptability of a given capital allocation vector. Specifically, when we decompose an aggregation function into attributed pieces for each institution, it is not guaranteed that the acceptability of all banks will lead to systemic acceptability; we refer the interested reader to Example~\ref{ex:non-acceptable} for a simple example \emph{without} systemic acceptability. Proposition~\ref{prop:acceptable} provides a simple sufficient condition for acceptability. Recall that Assumption~\ref{asmp:coherent} is imposed herein and for the remainder of this work.

\begin{proposition}\label{prop:acceptable}
Let $(\Lambda_i)_{i\in[N]}$ be a decomposition of $\Lambda$. Consider a random shock $X \in L_\Lambda^\infty(\R^N)$ and a capital vector $m\in D(X)$. If $\Lambda_i(X,m) \in \A$ for every bank $i\in[N]$, then $\Lambda(X,m) \in \A$. 
\end{proposition}

\begin{proof}
Since $m\in D(X)$, we have $\Lambda(X,m)\in L^\infty(\R)$ by Lemma~\ref{asmp:domain}. Suppose that $\Lambda_i(X,m) \in \A$ for every $i\in[N]$. Since the acceptance set $\A$ is a convex cone, we have $\sum_{i\in[N]}\Lambda_i(X,m)\in\A$. Finally, since the underlying risk measure $\rho$ is monotone and $\sum_{i\in[N]}\Lambda_i(X,m)\leq \Lambda(X,m)$, we obtain $\Lambda(X,m)\in\A$.
\end{proof}

\begin{example}\label{ex:non-acceptable}
Proposition~\ref{prop:acceptable} guarantees systemic acceptability by imposing a coherence assumption on the acceptance set $\A$. In this example, we consider a convex but \emph{not} coherent risk measure to demonstrate a case where $\Lambda_i(X,m) \in \A$ for each bank $i\in [N]$ but $\Lambda(X,m) \not\in \A$.
Specifically, we take the entropic risk measure with unit risk aversion and the summation aggregation function over $N = 2$ banks, i.e., $\A = \{Y \in L^\infty(\R) \; | \; \E[1 - e^{-Y}] \geq 0\}$ and $\Lambda(x,m) = \sum_{i \in [N]} (x_i + m_i)$ for each $(x,m) \in \R^N \times \R^N$.
Herein, we consider the natural decomposition of the aggregation function given by $\Lambda_i(x,m) = x_i + m_i$ for each $i\in[N]$ and $(x,m) \in \R^N \times \R^N$.
Finally, we set $X = (X_1,X_2)^\T \in L^\infty(\R^N)$ with $X_1 = X_2$ and $\mathbb{P}\{X_1 = -\frac12\} = \mathbb{P}\{X_1=2\}=\frac12$. Taking $\bar{m} = (0,0)^\T$, it is clear that $\E[1 - e^{-\Lambda_1(X,\bar{m})}] = \E[1 - e^{-\Lambda_2(X,\bar{m})}] \approx 0.108 \geq 0$ whereas $\E[1 - e^{-\Lambda(X,\bar{m})}] \approx -0.368 < 0$.
\end{example}

\begin{corollary}\label{cor:acceptable}
Let $(\Lambda_i)_{i\in[N]}$ be a decomposition of $\Lambda$. Then, every $(\A,(\Lambda_i)_{i\in[N]})$-Nash allocation rule is acceptable.
\end{corollary}

\begin{proof}
Let $r$ be an $(\A,(\Lambda_i)_{i\in[N]})$-Nash allocation rule. Let $X\in L^\infty_\Lambda(\R^N)$ and $i\in[N]$. Then, there exists a decreasing sequence $(m_i^n)_{n\in\N}$ in $\R$ that converges to $r_i(X)$ such that
\[
\rho(\Lambda_i(X,(m^n_i,r_{-i}(X)))) \leq 0
\]
for each $n\in\N$. 
By the monotonicity and continuity of $\Lambda_i$ on $\dom\Lambda_i$, we have
\[
\lim_{n\rightarrow\infty} \Lambda_i(X,(m^n,r_{-i}(X)))
%=\limsup_{n\rightarrow\infty} \Lambda_i(X,(m^n,r_{-i}(X)))
\leq \Lambda_i(X,r(X))
\]
Then, since $\Lambda_i(X,(m^n,r_{-i}(X)))\downarrow \lim_{n\rightarrow\infty} \Lambda_i(X,(m^n,r_{-i}(X)))$, we have
\[
\rho(\Lambda_i(X,r(X)))\leq \rho\of{\lim_{n\rightarrow\infty} \Lambda_i(X,(m^n,r_{-i}(X)))}=\lim_{n\rightarrow\infty}\rho( \Lambda_i(X,(m^n,r_{-i}(X))))\leq 0,
\]
where we use the fact that $\rho$ is continuous from above (see Assumption~\ref{asmp:coherent}) and monotone. Therefore, $\Lambda_i(X,r(X))\in\A$. Following Proposition~\ref{prop:acceptable}, we obtain $\Lambda(X,r(X)) \in \A$ so that $r$ is acceptable.
\end{proof}

We wish to augment Proposition~\ref{prop:acceptable} to consider allocations when the joint distribution of the stress scenario $X$ may be unknown. Notably, in practice, estimating marginal distributions is a much simpler task than the full joint distribution, especially as the number of banks $N$ grows. While distribution estimation is not the focus of this work, the practicality of computing meaningful acceptable systemic allocations would require such considerations. The following result provides a method to find an acceptable allocation using only the marginal distributions when $\A$ is the acceptance set of a coherent optimized certainty equivalent risk measure; this result is adapted from \citet[Lemma D.12]{banerjee2022pricing}.

\begin{corollary}\label{cor:comonotonic}
Suppose that $\A$ is the acceptance set of an optimized certainty equivalent risk measure $\rho$, i.e., 
\[
\rho(Y) := \inf_{m \in \R} (m - \E[u(Y+m)])
\]
for $Y \in L^\infty(\R)$, where $u\colon \R \to \R$ is defined by
\[
u(t) = \begin{cases} \gamma_2 t &\text{if } t \leq 0, \\ \gamma_1 t &\text{if } t > 0,\end{cases}
\]
and $\gamma_2 \geq 1 \geq \gamma_1 \geq 0$. Moreover, suppose that, for every $m\in \R^N$, the function $x \mapsto \Lambda(x,m)$ is submodular on $\R^N$, i.e.,
\[
\Lambda(x\vee y,m)+\Lambda(x\wedge y,m)\leq \Lambda(x,m)+\Lambda(y,m)
\]
for every $x,y\in\R^N$, where $\vee$ and $\wedge$ denote coordinate-wise maximum and minimum operations, respectively. Let $(\Lambda_i)_{i\in[N]}$ be a decomposition of $\Lambda$. Let $X \in L^\infty(\R^N)$ be a random shock and denote by $Z$ a comonotonic copula of $X$, i.e.,
\[
Z \coloneqq \of{q_{X_1}(U),\ldots,q_{X_N}(U)}^{\T}
\]
for some random variable $U$ that is uniformly distributed on $[0,1]$, where $q_{X_1},\ldots,q_{X_N}$ are the right-continuous quantile functions of $X_1,\ldots,X_N$, respectively. Let $m\in D(Z)$.
If $\Lambda_i(Z,m) \in \A$ for every bank $i\in[N]$, then $\Lambda(X,m) \in \A$.
\end{corollary}

\begin{proof}
The function $u$ is a utility function in the sense of \citet[Definition~2.1]{ben2007old} since it is concave, continuous, increasing, finitely valued with $u(0)=0$, and $1$ is a subgradient of $u$ at zero. Hence, $\rho$ is a convex risk measure by \citet[Theorem~2.1]{ben2007old}. Moreover, the positive homogeneity of $u$ implies that $\rho$ is also coherent. Suppose that $\Lambda_i(Z,m)\in\A$ for every $i\in[N]$. First, we utilize Proposition~\ref{prop:acceptable} to conclude that $\Lambda(Z,m) \in \A$. We claim that $\Lambda(X,m)\in\A$ as well. By \citet[Theorem 9.A.21]{shaked2007stochastic}, we have $\E[\phi(X)] \geq \E[\phi(Z)]$ for every submodular function $\phi\colon\R^N \to \R$. 
In fact $x \mapsto u(\Lambda(x,m)+y)$ is submodular as it is the composition of a concave and increasing function ($u(\cdot+y)$) and a submodular and increasing function ($\Lambda(\cdot,m)$) for every $m \in D(X)$ and $y \in \R$. 
Therefore, for every $y \in \R$, we have $y - \E[u(\Lambda(X,m)+y)] \leq y - \E[u(\Lambda(Z,m)+y)]$. 
In particular, by the construction of the optimized certainty equivalent, it immediately follows that $\rho(\Lambda(X,m)) \leq \rho(\Lambda(Z,m)) \leq 0$, i.e., $\Lambda(X,m) \in \A$.
\end{proof}

\subsection{Existence and Uniqueness}\label{sec:nash-exist}

In the above section, we take the existence of a Nash allocation rule for granted to prove that, for coherent risk measures that are continuous from above, any Nash allocation leads to an acceptable capital requirement. In the next two theorems, we consider the existence and uniqueness of Nash allocation rules based on the properties of the decomposition of the aggregation function.

\begin{theorem}[Existence]\label{thm:exist-unique}
Let $(\Lambda_i)_{i \in [N]}$ be a decomposition of $\Lambda$. Then, there exists a $(\A,(\Lambda_i)_{i \in [N]})$-Nash allocation rule $r\colon L_\Lambda^\infty(\R^N) \to \R^N$.
\end{theorem}

\begin{proof} %\notiz{Update proof}
Let $X \in L_{\Lambda}^\infty(\R^N)$ be a random shock. For each $m\in\R^N$ and $i\in[N]$, we define
\begin{equation}\label{eq:phi}
\phi^X_i(m) \coloneqq \inf\{r_i \in \R \; | \; \Lambda_i(X,(r_i,m_{-i})) \in \A\}.
\end{equation}
Clearly, a vector $r(X)\in\R^N$ satisfies the Nash allocation property in \eqref{eq:NE} if and only if it is a fixed point of $\phi^X$, i.e., $\phi^X(r(X))=r(X)$. We will show the existence of such a fixed point in several steps below. We will first construct a compact hyperrectangle and show that $\phi^X$ maps this hyperrectangle into itself. Then, we will show that all fixed points of $\phi^X$ belong to this hyperrectangle.

\begin{itemize}
\item \textbf{Construction of the hyperrectangle parameters:} Recall $D(X)=D(\essinf X)\neq\emptyset$ by Lemma~\ref{asmp:domain}(i). By self-feasibility, we can find $\bar{r}\in D(\essinf X)$ such that
\begin{equation}\label{eq:feas}
\Lambda_i(\essinf X,(\bar{r}_i,m_{-i})) > 0
\end{equation}
for every $m\in D(\essinf X)$ and $i\in[N]$. 
On the other hand, since $\bar{r}\in D(\essinf X)\subseteq D(\esssup X)$, we also have $D(\esssup X)\neq\emptyset$. Then, by self-infeasibility, we can find $\underline{r}\in D(\esssup X)$ such that
\begin{equation}\label{eq:infeas}
\Lambda_i(\esssup X,(\underline{r}_i-\varepsilon,m_{-i}))<0
\end{equation}
for every $m\in D(\esssup X)$, $i\in[N]$, and $\varepsilon>0$.

\item \textbf{Construction of the hyperrectangle:} Let us fix $m\in D(X)=D(\essinf X)$ and $i\in[N]$. Then, by monotonicity and \eqref{eq:feas},
\begin{align*}
\phi^X_i(m)&\leq \inf\{r_i\in\R\; | \; \Lambda_i(X,(r_i,m_{-i}))\geq 0\}\\
&\leq \inf\{r_i\in\R\; | \; \Lambda_i(\essinf X,(r_i,m_{-i}))\geq 0\}\leq \bar{r}_i.
\end{align*}
Let $\varepsilon>0$. By monotonicity and \eqref{eq:infeas},
\[
\Lambda_i(X,(\underline{r}_i-\varepsilon,m_{-i}))\leq \Lambda_i(\esssup X,(\underline{r}_i-\varepsilon,m_{-i}))<0 \text{ a.s.}
\]
Then, by the properties of $\rho$ as a coherent risk measure, we have
\[
\rho(\Lambda_i(X,(\underline{r}_{i}-\varepsilon,m_{-i})))\geq \rho(\Lambda_i(\esssup X,(\underline{r}_{i}-\varepsilon,m_{-i})))=-\Lambda_i(\esssup X,(\underline{r}_{i}-\varepsilon,m_{-i}))>0.
\]
Hence, $\Lambda_i(X,(\underline{r}_i-\varepsilon,m_{-i}))\notin \A$ and, by monotonicity, we have $\phi^X_i(m)\geq \underline{r}_i-\varepsilon$. Since $\varepsilon>0$ is arbitrary, we conclude that $\phi^X_i(m)\geq \underline{r}_i$. Therefore, $\phi^X(m)\in [\underline{r},\bar{r}]$ for every $m\in D(X)$. In particular, since $\bar{r}\in D(\essinf X)=D(X)$, we have $\phi^X(\bar{r})\in [\underline{r},\bar{r}]$. On the other hand, since $\phi^X$ is decreasing by construction, we have $m\leq \bar{r}$ implies $\phi^X(\bar{r})\leq \phi^X(m)$. We will work with the hyperrectangle $[\phi^X(\bar{r}),\bar{r}]$ to complete the proof of existence.

\item \textbf{Self-mapping property of $\phi^X$ on the hyperrectangle:} We claim that $\phi^X(\bar{r})\in D(X)$ so that $[\phi^X(\bar{r}),\bar{r}]\subseteq D(X)$. The rectangular structure of $D(X)=D(\essinf X)$ ensures that $D(X)=\bigtimes_{i=1}^N A_i$, where, for each $i\in[N]$, we have $A_i=\R$ or $A_i=[\bar{m}_i,+\infty)$ for some $\bar{m}_i\in\R$. Let us fix $i\in [N]$. We may write
\[
\phi^X_i(\bar{r})=\inf\{r_i\in A_i\mid \Lambda_i(X,(r_i,\bar{r}_{-i}))\in\A\}.
\]
Since $\phi^X_i(\bar{r})\geq \underline{r}_i>-\infty$ and $\rho$ is continuous from above, the above infimum is indeed a minimum so that $\phi^X_i(\bar{r})\in A_i$ as well. Hence, $\phi^X(\bar{r})\in D(\essinf X)=D(X)$. This finishes the proof of the claim.
Therefore, $\phi^X$ maps the compact hyperrectangle $[\phi^X(\bar{r}),\bar{r}]$ of $D(X)$ into itself.

\item \textbf{Inclusion of fixed points in the hyperrectangle:} Let $m\in\R^N$ be an arbitrary fixed point of $\phi^X$. Then, as argued for $\phi^X(\bar{r})$ above, the definition of $\phi^X(m)$ and the finiteness of $m$ guarantee that $\phi^X(m)\in D(X)$. Hence, $m=\phi^X(m)\in [\underline{r},\bar{r}]$. Then, since $\phi^X$ is decreasing, we have $\phi^X(\bar{r})\leq \phi^X(m)=m$. Hence, $m\in [\phi^X(\bar{r}),\bar{r}]$. This shows that every fixed point of $\phi^X$ belongs to $[\phi^X(\bar{r}),\bar{r}]$.

\item \textbf{Continuity of $\phi^X$ relative to the hyperrectangle:} Next, we prove that the restricted function $\phi^X\colon [\phi^X(\bar{r}),\bar{r}]\to [\phi^X(\bar{r}),\bar{r}]$ is continuous (relative to $[\phi^X(\bar{r}),\bar{r}]$). Let $i\in[N]$. Then, for each $m\in [\phi^X(\bar{r}),\bar{r}]$, we may write $\phi^X_i(m)=\inf\Phi^X_i(m)$, where
\[
\Phi^X_i(m):=\{r_i\in [\phi^X_i(\bar{r}),\bar{r}_i] \;|\; f_i^X(r_i,m_{-i})\leq 0\},\quad f_i^X(m):=\rho(\Lambda_i(X,m)).
\]
We claim that $f^X_i$ is a continuous convex function relative to $[\phi^X(\bar{r}),\bar{r}]$. By the convexity of $\rho$ and the concavity of $\Lambda_i$, the convexity of $f^X_i$ is immediate. Repeating the arguments above, it is easy to see that $f^X_i$ is also real-valued. Let $m\in [\phi^X(\bar{r}),\bar{r}]$ and take a sequence $(m^n)_{n\in\N}$ in $[\phi^X(\bar{r}),\bar{r}]$ that converges to $m$. Then, by monotonicity and continuity, $(\Lambda_i(X,m^n))_{n\in\N}$ is a bounded sequence in $L^\infty(\R)$ that converges to $\Lambda_i(X,m)$ a.s. By \citet[Theorem~4.33]{fs:sf}, Assumption~\ref{asmp:coherent} guarantees that $\rho$ has the Fatou property so that
\[
f^X_i(m)=\rho(\Lambda_i(X,m))\leq \liminf_{n\rightarrow\infty}\rho(\Lambda_i(X,m^n))=\liminf_{n\rightarrow\infty}f^X_i(m^n).
\]
Hence, $f^X_i$ is lower semicontinuous relative to $[\phi^X(\bar{r}),\bar{r}]$ at $m$. As $[\phi^X(\bar{r}),\bar{r}]$ is a polyhedral set, \citet[Theorem~10.2]{rockafellar} implies that $f^X_i$ is continuous relative to $[\phi^X(\bar{r}),\bar{r}]$, as claimed. In particular, $\Phi^X_i(m)$ is a nonempty compact set and $\phi^X_i(m)=\min\Phi^X_i(m)$ for each $m\in [\phi^X(\bar{r}),\bar{r}]$. Once we prove that the set-valued function $\Phi^X_i$ is upper and lower hemicontinuous relative to $[\phi^X(\bar{r}),\bar{r}]$, Berge maximum theorem (see \citet[Theorem~17.31]{aliprantis-border}) guarantees the continuity of $\phi^X_i$ relative to $[\phi^X(\bar{r}),\bar{r}]$.

Let us fix a capital level $m\in [\phi^X(\bar{r}),\bar{r}]$.
\begin{itemize}
\item \textbf{Upper hemicontinuity of $\phi^X$:} To prove relative upper hemicontinuity at $m$, let $(m^n)_{n\in\N}$ be a sequence in $[\phi^X(\bar{r}),\bar{r}]$ that converges to $m$. For each $n\in\N$, let $r_i^n\in \Phi^X_i(m^n)$. Since $(r_i^n)_{n\in\N}$ is a bounded sequence, it has a convergent subsequence $(r_i^{n_k})_{k\in\N}$. Let $r_i\in [\phi^X_i(\bar{r}),\bar{r}_i]$ be its limit. Then, by the continuity of $f^X_i$ established above, 
\[
f^X_i(r_i,m_{-i})= \lim_{k\rightarrow\infty} f^X_i(r_i^{n_k},m^{n_k}_{-i})\leq 0. 
\]
Hence, $r_i\in\Phi^X_i(m)$, which proves that $\phi^X_i$ is upper hemicontinuous relative to $[\phi^X(\bar{r}),\bar{r}]$ at $m$ by \citet[Theorem~17.20]{aliprantis-border}.

\item \textbf{Lower hemicontinuity of $\phi^X$:} To prove relative lower hemicontinuity at $m$, first note that, by monotonicity and \eqref{eq:feas},
\begin{align*}
f^X_i(\bar{r}_i,m_{-i})&=\rho(\Lambda_i(X, (\bar{r}_i,m_{-i})))\\
&\leq \rho(\Lambda_i(\essinf X, (\bar{r}_i,m_{-i})))=-\Lambda_i(\essinf X, (\bar{r}_i,m_{-i}))<0.
\end{align*}
Since $f^X_i(\cdot,m_{-i})$ is a continuous function relative to $[\phi^X_i(\bar{r}),\bar{r}_i]$, there exists $r^m_i\in (\phi^X_i(\bar{r}),\bar{r}_i)$ such that $f^X_i(r^m_i,m_{-i})<0$. Note that the extension of the function $f^X_i(\cdot,m_{-i})$ to $\R$ by setting its value as $+\infty$ outside $[\phi^X_i(\bar{r}),\bar{r}_i]$ is a convex function. Then, by \citet[Proposition~2.34]{rockafellarwets}, we obtain
\[
\interior \Phi^X_i(m)=\{r_i\in (\phi^X_i(\bar{r}),\bar{r}_i)\mid f^X_i(r_i,m_{-i})<0\}.
\]
Hence, $r_i^{m}\in \interior \Phi^X_i(m)$ so that $\interior \Phi^X_i(m)\neq\emptyset$.

Let us fix an arbitrary point $r_i\in \interior \Phi^X_i(m)$. In particular, $f^X_i(r_i,m_{-i})<0$. By the continuity of $f^X_i$ relative to $[\phi^X(\bar{r}),\bar{r}]$, there exist relatively open sets $U\subseteq [\phi^X_i(\bar{r}),\bar{r}_i]$ and $V\subseteq [\phi^X(\bar{r}),\bar{r}]$ such that $(r_i,m)\in U\times V$ and
\begin{equation}\label{eq:intpt}
f^X_i(\hat{r}_i,\hat{m}_{-i})<0
\end{equation}
for every $(\hat{r}_i,\hat{m})\in U\times V$. In particular, $V=\tilde{V}\cap [\phi^X(\bar{r}),\bar{r}]$ for some open neighborhood $\tilde{V}\subseteq \R^N$ of $m$. On the other hand, since $r_i\in \interior \Phi^X_i(m)\subseteq (\phi^X_i(\bar{r}),\bar{r}_i)$, we have $U\subseteq (\phi^X_i(\bar{r}),\bar{r}_i)$ so that $U$ is open in $\R$. Hence, $\tilde{V}\times U$ is an open neighborhood of $(m,r_i)$ and, by \eqref{eq:intpt}, we have $(\tilde{V}\cap [\phi^X(\bar{r}),\bar{r}]) \times U= V\times U\subseteq \gr \Phi^X_i$, where $\gr \Phi^X_i:=\{(\hat{m},\hat{r}_i)\mid \hat{r}_i\in\Phi^X_i(\hat{m})\}$. Since $\Phi^X_i$ has convex (interval) values, by \citet[Theorem~5.9(a)]{rockafellarwets}, we conclude that $\Phi^X_i$ is lower hemicontinuous relative to $[\phi^X(\bar{r}),\bar{r}]$ at $m$.
\end{itemize}
Hence, the proof of the continuity of $\phi^X_i$ (hence, of $\phi^X$) relative to $[\phi^X(\bar{r}),\bar{r}]$ is complete.

\item \textbf{Existence of a fixed point:} Since $[\phi^X(\bar{r}),\bar{r}]$ is a nonempty convex compact set in $\R^N$, by Brouwer fixed point theorem (see \citet[Theorem~17.56]{aliprantis-border}), $\phi^X$ has a fixed point $r(X)\in [\phi^X(\bar{r}),\bar{r}]$.
\end{itemize}
Hence, the proof of the existence of a Nash allocation rule $r\colon L^\infty_\Lambda(\R^N)\to\R^N$ is complete.
\end{proof}

To obtain the uniqueness of a Nash allocation rule in the next theorem, we will assume that the decomposition of the aggregation function is self-preferential.

\begin{theorem}[Uniqueness]\label{thm:exist-unique:2} 
Let $(\Lambda_i)_{i \in [N]}$ be a decomposition of $\Lambda$. Suppose that $(\Lambda_i)_{i \in [N]}$ is self-preferential with constant $L < \frac{1}{N-1}$. Then, there exists a unique $(\A,(\Lambda_i)_{i \in [N]})$-Nash allocation rule.
\end{theorem}

\begin{proof}
Let $X\in L^\infty(\R^N)$ be a random shock and let $\phi^X$ be the function defined by \eqref{eq:phi} in the proof of Theorem~\ref{thm:exist-unique}. We proceed in two steps.
\begin{itemize}
\item \textbf{Lipschitz continuity of $\phi^X$:} We first prove that $\phi^X$ is Lipschitz continuous on $D(X)$ with constant $L(N-1)$ with respect to the $\ell^\infty$-norm on $\R^N$ by assuming that $(\Lambda_i)_{i \in [N]}$ is self-preferential with constant $L\in [0,1)$. Let us fix two capital levels $m,m^\prime \in D(X)$ and a bank $i \in [N]$. Without loss of generality, we assume that $\phi^X_i(m) \leq \phi^X_i(m^\prime)$. We claim that
\begin{equation}\label{eq:lipschitz}
\phi^{X}_i(m^\prime) \leq \phi^{X}_i(m) + L\sum_{j \in [N]\setminus\{i\}} |m_j - m^\prime_j|.
\end{equation}
To that end, let us fix $r_i\in\R$ such that $\Lambda_i(X,(r_i,m_{-i}))\in\A$ and show that
\[
\phi^{X}_i(m^\prime) \leq r_i + L\sum_{j \in [N]\setminus\{i\}} |m_j - m^\prime_j|.
\]
By the definition of $\phi^X_i(m^\prime)$, this inequality follows once we establish that
\begin{equation}\label{eq:lipschitz2}
\Lambda_i\of{X,\of{r_i+L\sum_{j\in[N]\setminus\{i\}} |m_j-m^\prime_j| , m^\prime_{-i}}} \in \A.
\end{equation}
Let us define $m^{\prime\prime}\in\R^N$ by $m^{\prime\prime}_j:=m^\prime_j+|m_j-m^\prime_j|$ for each $j\in[N]$. Using the monotonicity and self-preferential property (for each $j\in[N]\setminus\{i\}$) of $\Lambda_i$, we get
\[
\Lambda_i(X,(r_i,m_{-i})) \leq \Lambda_i(X,(r_i,m^{\prime\prime}_{-i}))
    \leq \Lambda_i\of{X,\of{r_i+L\sum_{j \in[N]\setminus\{i\}} |m_j-m^\prime_j|,m^\prime_{-i}}}.
\]
Since we have $\Lambda_i(X,(r_i,m_{-i}))\in\A$ by assumption, we get \eqref{eq:lipschitz2} and hence \eqref{eq:lipschitz}.
As the latter holds for every $i\in[N]$, it follows that
\begin{align*}
\max_{i\in[N]} |\phi^{X}_i(m)-\phi^{X}_i(m^\prime)| &\leq L \max_{i\in[N]} \sum_{j\in[N]\setminus\{i\}} |m_j-m^\prime_j|\\ 
\qquad &\leq (N-1) L \max_{i\in[N]} \max_{j \in[N]\setminus\{i\}} |m_j-m^\prime_j| = (N-1) L \max_{i \in [N]} |m_i - m^\prime_i|
\end{align*}
for every $m,m^\prime\in D(X)$. Hence, $\phi^X$ is Lipschitz continuous on $D(X)$ with constant $L(N-1)$. 

\item \textbf{Uniqueness of a fixed point:} To finish the uniqueness proof, let us suppose that $L<\frac{1}{N-1}$. Then, $\phi^X$ is Lipschitz continuous on $D(X)$ with constant $L(N-1)<1$ so that it is a strict contraction. Hence, it has a unique fixed point $r(X)$ in the complete metric space $[\phi^X(\bar{r}),\bar{r}]\subseteq D(X)$ by Banach fixed point theorem, where $[\phi^X(\bar{r}),\bar{r}]$ is the compact hyperrectangle constructed in the proof of Theorem~\ref{thm:exist-unique}; see \eqref{eq:feas}. As argued in that proof, every fixed point of $\phi^X$ belongs to this hyperrectangle. Hence, $\phi^X$ has a unique fixed point.
\end{itemize}
Thus, the uniqueness of a Nash allocation rule follows.
\end{proof}

\begin{remark}\label{rem:N=2}
Following Theorem~\ref{thm:exist-unique:2}, in the setting of $N = 2$ banks, there always exists a unique $(\A,(\Lambda_1,\Lambda_2))$-Nash allocation rule for a self-preferential decomposition $(\Lambda_1,\Lambda_2)$ of $\Lambda$. However, as the number of banks $N$ in the system grow, the uniqueness condition of Theorem~\ref{thm:exist-unique:2} becomes more restrictive.
\end{remark}

\subsection{Application to Single-Element Aggregation Functions}\label{sec:single-element}

In this section, we consider applications of Theorems~\ref{thm:exist-unique},~\ref{thm:exist-unique:2} to the specific systemic risk measure structures of Example~\ref{ex:syst}. In so doing, we propose decomposition structures that can be used for different systemic risk measure structures based on single-element aggregation functions, e.g., the ones provided in Example~\ref{ex:se-aggregation}.

Throughout this section, we fix a single-element aggregation function $\bar\Lambda\colon \R^N\to \R\cup\{-\infty\}$. Similar to \eqref{eq:decomposition}, we assume that $\bar\Lambda$ has the form
\begin{equation}\label{eq:singledec}
\bar\Lambda(x)=\sum_{i\in[N]}\bar\Lambda_i(x)
\end{equation}
for every $x\in \R^N$, where $\bar\Lambda_1,\ldots,\bar\Lambda_N\colon \R^N\to \R\cup \{-\infty\}$ are single-element aggregations functions with common domain, i.e., $\dom\bar\Lambda_i=\dom\bar\Lambda$ for each $i\in[N]$. We refer to $(\bar\Lambda_i)_{i\in[N]}$ as a \textbf{single-element decomposition} of $\bar\Lambda$.

\begin{example}\label{ex:decomposition}
Let us revisit the single-element aggregation functions in Example~\ref{ex:se-aggregation} and propose single-element decompositions for them.
\begin{enumerate}
\item\label{util2} \textbf{Sum of univariate utilities:} For the single-element aggregation function $\bar\Lambda^U$ in Example~\ref{ex:se-aggregation}\eqref{util} and Example~\ref{ex:sen-ut}, for each $i \in [N]$ and $x\in\R^N$, let $\bar\Lambda_i^U(x)\coloneqq u_i(x_i)$ if $x\in \dom\bar\Lambda^U$ and $\bar\Lambda_i^U(x)\coloneqq -\infty$ otherwise. Clearly, $\bar\Lambda_1^U,\ldots,\bar\Lambda_N^U$ are single-element aggregation functions with common domain $\dom\bar\Lambda^U$. Moreover, we have $\sum_{i\in [N]}\bar\Lambda^U_i(x)=\bar\Lambda^U(x)$ for every $x\in\R^N$. Hence, $(\bar\Lambda^U_i)_{i\in[N]}$ is a single-element decomposition of $\bar\Lambda^U$.
\item\label{mf2} \textbf{Mean-field utility:} Consider now the single-element aggregation function $\bar\Lambda^{MF}$ in Example~\ref{ex:se-aggregation}\eqref{mf}. Herein, we assume that $(u_i)_{i \in [N]},\bar u$ are twice continuously differentiable univariate utility functions with domain $\R_+$ for simplicity.
Following the decomposition approach in Remark~\ref{rem:select} and Appendix~\ref{app:nash-decomposition}, we set $\bar\lambda_i(x) := u_i(x_i)$ for every $x \in \R^N$ as with the sum of univariate utilities. However, in contrast to that simpler aggregation function, with the systemic externalities, we retain a remainder term $\bar\Delta(x) := \bar u(\frac{1}{N} \sum_{i \in [N]} x_i)$ for this aggregation function. As proven in Appendix~\ref{app:mf-decomp}, the minimal decomposition, i.e., the one with the minimal self-preferential constant is $\bar\delta_i(x) := \frac{1}{N}\bar u(\frac{1}{N} \sum_{j \in [N]} x_j)$ for all banks $i$. In this way, we can consider the single-element decomposition $\bar\Lambda_i^{MF}(x) := u_i(x_i) + \frac{1}{N}\bar u(\frac{1}{N} \sum_{j \in [N]} x_j)$ for every bank $i$ for the purposes of this section.
\end{enumerate}
\end{example}

\subsubsection{Insensitive Systemic Risk Measures}\label{sec:nash-insensitive}
Throughout this section, we consider the insensitive systemic risk measures discussed in Example~\ref{ex:syst}. To that end, corresponding to the single-element aggregation function $\bar\Lambda\colon\R^N\to\R\cup\{-\infty\}$, we define the insensitive aggregation function $\Lambda^I\colon\R^N\times\R^N\to\R\cup\{-\infty\}$ by \eqref{eq:lambda-ins}:
\[
\Lambda^I(x,m)\coloneqq \bar\Lambda(x)+\sum_{i\in [N]}m_i
\]
for each $x,m\in\R^N$. Recall that we have $L^\infty_{\Lambda^I}(\R^N)=L^\infty(\dom\bar\Lambda)$ in this setting. 
The following corollary characterizes the \emph{unique} Nash allocation rule that exists for any insensitive systemic risk measure.

\begin{corollary}\label{cor:insensitive}
Let $(\bar\Lambda_i)_{i\in[N]}$ be a single-element decomposition of $\bar\Lambda$.
For each $i\in[N]$, define a function $\Lambda^I_i\colon \R^N\times\R^N\to\R\cup\{-\infty\}$ by
\[
\Lambda^I_i(x,m)\coloneqq \bar\Lambda_i(x)+m_i
\]
for every $x,m\in \R^N$. Then, $(\Lambda_i^I)_{i\in[N]}$ is a self-preferential decomposition of $\Lambda^I$ with $L = 0$.
Moreover, the unique $(\A,(\Lambda_i^I)_{i\in [N]})$-Nash allocation rule $r^I\colon L^\infty(\dom\bar\Lambda)\to\R^N$ is given by
\[
r^I(X):=\of{\rho(\bar\Lambda_1(X)),\ldots,\rho(\bar\Lambda_N(X))}^\T
\]
for every $X\in L^\infty(\dom\bar\Lambda)$.
\end{corollary}

\begin{proof}
It is easy to verify that $(\Lambda_i^I)_{i\in[N]}$ is a self-preferential decomposition of $\Lambda^I$ with $L = 0$. Then, by Theorem~\ref{thm:exist-unique:2}, there exists a unique Nash allocation rule $r^I\colon L^\infty(\dom\bar\Lambda) \to \R^N$.
It remains to show that $r^I$ has the form provided in the statement of the corollary.
Let $X\in L^\infty(\dom\bar\Lambda)$. Then, a $(\A,(\Lambda_i^I)_{i\in [N]})$-Nash allocation rule $r^I\colon L^\infty(\dom\bar\Lambda)\to\R^N$ must satisfy
\begin{align*}
r^I_i(X)&=\inf\cb{m_i\in\R\; | \; \Lambda^I_i(X,(m_i,r^I_{-i}(X)))\in\A}\\
&=\inf\cb{m_i\in\R\; | \; \bar\Lambda_i(X)+m_i\in\A}=\rho(\bar\Lambda_i(X)),
\end{align*}
which follows trivially from the primal representation of the risk measure $\rho$ (Proposition~\ref{prop:riskmsr-primal}) and the definition of $\Lambda_i^I$.
\end{proof}

\subsubsection{Sensitive Systemic Risk Measures}\label{sec:nash-sensitive}

We now turn our attention to the more complex interconnections within the sensitive systemic risk measures, that is, we consider the sensitive aggregation function $\Lambda^S\colon\R^N\times\R^N\to\R\cup\{-\infty\}$ defined by \eqref{eq:lambda-sen}:
\[
\Lambda^S(x,m)\coloneqq \bar\Lambda(x+m)
\]
for each $x,m\in \R^N$. Recall that we have $L^\infty_{\Lambda^S}(\R^N)=L^\infty(\R^N)$ in this setting.

\begin{corollary}\label{cor:sensitive}
Let $(\bar\Lambda_i)_{i \in [N]}$ be a single-element decomposition of $\bar\Lambda$ that satisfies the following properties:
\begin{enumerate}
\item \textbf{Self-feasible:} There exists $\bar{x} \in \dom\bar\Lambda$ such that $\bar\Lambda_i(\bar{x}_i,x_{-i}) > 0$ for every $x \in \dom\bar\Lambda$ and $i \in [N]$.
\item \textbf{Self-infeasible:} There exists $\underline{x} \in \dom\bar\Lambda$ such that $-\infty \leq \bar\Lambda_i(\underline{x}_i-\varepsilon,x_{-i}) < 0$ for every $x \in \dom\bar\Lambda$, $i \in [N]$, and $\varepsilon > 0$.
\end{enumerate}
For each $i\in [N]$, define a function $\Lambda^S_i\colon \R^N\times\R^N\to \R\cup\{-\infty\}$ by
\begin{equation}\label{eq:sen-dec}
\Lambda^S_i(x,m)\coloneqq \bar\Lambda_i(x+m)
\end{equation}
for every $x,m\in\R^N$. Then, there exists a Nash allocation rule $r^S\colon L^\infty(\R^N)\to\R^N$.
Moreover, if there exists a constant $L\in [0,1)$ such that, for every $i,j\in[N]$ with $i \neq j$, $x\in\dom\bar\Lambda$, and $\alpha\geq 0$, it holds $\bar\Lambda_i(x+L\alpha e_i) \geq \bar\Lambda_i(x+\alpha e_j)$, then $(\Lambda^S_i)_{i \in [N]}$ is a self-preferential decomposition of $\Lambda^S$ and admits a unique Nash allocation rule $r^S$ if $L < \frac{1}{N-1}$.
\end{corollary}
\begin{proof}
Immediately, we observe that $\Lambda^S_1,\ldots,\Lambda^S_N$ are aggregation functions, and $(\Lambda^S_i)_{i\in[N]}$ satisfies \eqref{eq:decomposition} and the consistent domain property. For existence, following Theorem~\ref{thm:exist-unique}, it remains to show that self-feasibility and self-infeasibility hold.
\begin{itemize}
\item \textbf{Self-feasible:} Let $x \in \R^N$ with $D(x) \neq \emptyset$ and define $\bar{r} := \bar{x} - x$. We have $\bar{r}\in D(x)$ since, by construction, $x + \bar{r}= \bar{x} \in \dom\bar\Lambda$. Let $m \in D(x)$, i.e., $x + m \in \dom\bar\Lambda$. Then, for every $i\in[N]$, we have
\[
\Lambda_i^S(x,(\bar r_i , m_{-i})) = \bar\Lambda_i(x+(\bar r_i,m_{-i})) = \bar\Lambda_i(\bar x_i , x_{-i} + m_{-i}) > 0.
\]
\item \textbf{Self-infeasible:} Let $x \in \R^N$ with $D(x) \neq \emptyset$ and define $\underline{r}:= \underline{r} - x$. We have $\underline{r} \in D(x)$ since, by construction, $x + \underline{r} = \underline{x} \in \dom\bar\Lambda$. Let $m \in D(x)$, i.e., $x + m \in \dom\bar\Lambda$. Then, for every $i\in[N]$ and $\varepsilon > 0$, we have
\[
\Lambda_i^S(x,(\underline r_i-\varepsilon,m_{-i})) = \bar\Lambda_i(x+(\underline r_i-\varepsilon,m_{-i})) = \bar\Lambda_i(\underline x_i-\varepsilon , x_{-i} + m_{-i}) < 0.
\]
\end{itemize}
Then, by Theorem~\ref{thm:exist-unique}, the existence of a Nash allocation rule follows. Furthermore, it is easy to verify that the proposed condition $\bar\Lambda_i(x+L\alpha e_i) \geq \bar\Lambda_i(x+\alpha e_j)$ for every $i,j \in [N]$ with $i \neq j$, $x \in \dom\bar\Lambda$, and $\alpha \geq 0$ implies $(\Lambda_i^S)_{i \in [N]}$ is self-preferential with constant $L$. Then, by Theorem~\ref{thm:exist-unique:2}, the uniqueness of a Nash allocation rule follows.
\end{proof}

As opposed to the insensitive systemic risk measures (Section~\ref{sec:nash-insensitive}), the Nash allocations for sensitive systemic risk measures strongly depend on the aggregation function and its chosen decomposition. To demonstrate how the Nash allocations fit within the scheme of sensitive systemic risk measures, we will revisit the single-element aggregation functions and decompositions from Example~\ref{ex:decomposition}.

\begin{example}\label{ex:nash}
We consider the single-element aggregation functions and their associated decompositions in Example~\ref{ex:decomposition}. 
\begin{enumerate}
\item\label{ex:nash-utility} \textbf{Sum of univariate utilities:} Let us take $\bar\Lambda=\bar\Lambda^U$ and $\bar\Lambda_i=\bar\Lambda^U_i$ for each $i\in[N]$ as defined in Example~\ref{ex:decomposition}\eqref{util2}. Under the assumption that $\{x_i \in \dom u_i \; | \; u_i(x_i) \geq 0\} \not\in \{\emptyset,\R\}$ for each $i \in [N]$, it is easy to verify the properties in Corollary~\ref{cor:sensitive}. Hence, the corresponding collection $(\Lambda^{S}_i)_{i\in[N]}$ defined by \eqref{eq:sen-dec} is a self-preferential decomposition of the sensitive aggregation function $\Lambda^S$, as was also argued directly in Example~\ref{ex:sen-ut}. 
As with the insensitive systemic risk measures (Section~\ref{sec:nash-insensitive}), there are no interactions between the capital requirements of the banks under this decomposition. Therefore, similar to the proof of Corollary~\ref{cor:insensitive}, it follows that there exists a \emph{unique} $(\A,(\Lambda_i^S)_{i \in [N]})$-Nash allocation rule $r^S\colon L^\infty(\R^N)\to\R^N$ given by 
\[
r_i^S(X) = \inf\{m_i \in \R \; | \; u_i(X_i+m_i) \in \A\}
\]
for each $X\in L^\infty(\R^N)$ and $i\in[N]$.
\item\label{nash-mf} \textbf{Mean-field utility:} Let us take $\bar\Lambda = \bar\Lambda^{MF}$ and $\bar\Lambda_i = \bar\Lambda_i^{MF}$ for each $i \in [N]$ as defined in Example~\ref{ex:decomposition}\eqref{mf2}. Recall that, for simplicity, herein we consider $\dom\bar\Lambda = \R^N_+$. It is easy to verify that existence of a Nash allocation rule follows from Corollary~\ref{cor:sensitive} applies if there exists some $\bar x \in \R^N_+$ such that $u_i(\bar x_i) \geq -\frac{1}{N}\bar u(0)$, e.g., if $u_i$ is unbounded, for every $i \in [N]$. %In fact, by monotonicity of $u_i$, we guarantee the self-preferential property with $L \leq 1$. 
Because of the possible interactions between the firms in this case, it is no longer guaranteed that there exists a unique Nash allocation rule. However, despite the lack of a general theoretical guarantee, we find a unique Nash allocation for the following practical setting with $N = 2$. Let us take $u_1 = u_2 := \log(\varepsilon + \cdot)$ and $\bar u := \lambda\log(\varepsilon + \cdot)$ with domains $\R_+$ for some $\varepsilon \in (0,1)$ and $\lambda > 0$. Then, every $(\A,(\Lambda_i^S)_{i \in [N]})$-Nash allocation rule $r^S$ is uniquely evaluated at $X = 0$ a.s.\ as $r_1^S(0) = r_2^S(0) = 1-\varepsilon$.
\end{enumerate}
\end{example}

Despite the prior examples admitting unique Nash allocation rules, it is possible to construct self-preferential decompositions with multiple Nash allocation rules. In the following example we find a continuum of Nash allocations, which highlights the importance of choosing the decomposition appropriately. We note that though this counterexample satisfies the definition of a decomposition, it does not follow the construction of the minimal decomposition as presented in Appendix~\ref{app:nash-decomposition}.

\begin{example}\label{ex:nonunique}
Consider a system with $N = 3$ banks and a weighted summation aggregation function $\bar\Lambda\colon \R^N \to \R \cup \{-\infty\}$ defined by $\bar\Lambda(x) \coloneqq 1.5 x_1 + 1.5 x_2 + 2.25 x_3 - 3$ if $x \in \R^N_+$ and $\bar\Lambda(x)\coloneqq -\infty$ otherwise. This can be considered as a special case of sum of univariate utilities decomposed as $\bar\Lambda^U_1(x)=u_1(x_1)=1.5x_1-1$, $\bar\Lambda^U_2(x)=u_2(x_2)=1.5x_2-1$, $\bar\Lambda^U_3(x)=u_3(x_3)=2.25x_3-1$ for $x\in \R^N_+$ and $\bar\Lambda^U_1(x)=\bar\Lambda^U_2(x)=\bar\Lambda^U_3(x)=-\infty$ otherwise. This decomposition is associated with a unique associated Nash allocation rule as presented in Example~\ref{ex:nash}\eqref{ex:nash-utility}. Let us consider instead the alternative decomposition $\bar\Lambda_i(x) = a_i^\T x - 1$ for every $i\in[N]$ with $a_1 \coloneqq (1 , 0 , 0.75)^\T$, $a_2\coloneqq (0 , 1 , 0.75)^\T$, $a_3 \coloneqq (0.5 , 0.5 , 0.75)^\T$. In this case, for each $t\in [0,1]$, there exists a $(\A,(\Lambda_i^{S})_{i \in [N]})$-Nash allocation rule $r^S$ such that $r^S(0)=(1-t,1-t,4t/3)^\T$. Notably, in comparison to Example~\ref{ex:nash}\eqref{ex:nash-utility} above, there are infinitely many Nash allocation rules with this decomposition. As a consequence, it is clear that the choice of the decomposition can greatly influence the set of Nash allocation rules in practice.
\end{example}

\section{Case Study: Eisenberg-Noe Aggregation}\label{sec:EN}

For the remainder of the paper, we will consider an aggregation function derived from the Eisenberg-Noe clearing payments~\citep{eisenberg2001systemic}. In contrast to the aggregate utility functions which motivated our proposed structure, the Eisenberg-Noe aggregation function has complex dependencies between the banks within its natural decomposition. As far as the authors are aware, the Eisenberg-Noe aggregation function was first proposed in~\cite{feinstein2017measures}, though it is a special case of the optimization aggregation function in \citet[Example~4]{chen2013axiomatic} and is closely related to the contagion aggregation function in \citet[Example~7]{chen2013axiomatic}.

\subsection{Eisenberg-Noe Aggregation and Decomposition}\label{sec:EN-background}

Consider a system of $N$ banks that are interconnected via payment obligations. For each $i,j\in[N]$, we let $\bar p_{ij} \geq 0$ denote the obligations that bank $i$ owes to bank $j$; we set $\bar p_{ii} = 0$ to avoid self-dealing. In addition, each bank $i$ also has an obligation $\bar p_{i0} > 0$ external to the financial system; we often refer to these obligations as those due to ``society''. The Eisenberg-Noe clearing system~\citep{eisenberg2001systemic} assumes that these obligations are repaid proportionally. Therefore, from this construction, we set $\pi_{ij} = \bar p_{ij} / \sum_{k = 0}^N \bar p_{ik}$ to be the relative liability of bank $i\in [N]$ to entity (bank or society) $j \in [N]\cup\{0\}$. Assuming that the $N$ banks hold external assets worth $x\in \R^N_+$ and following the priority of debt over equity and limited liabilities, a \emph{clearing payment vector} $p(x)\in\R^N$ is defined as a solution of the fixed point equation
\begin{equation}\label{eq:EN}
p_i(x) = \min\left\{\sum_{j = 0}^N \bar p_{ij} \; , \; x_i + \sum_{j = 1}^N \pi_{ji} p_j(x)\right\}, \quad i \in[N].
\end{equation}
Immediately, by applying Tarski fixed point theorem, there exists greatest and least clearing payment vectors for $x \in \R^N_+$. Moreover, as provided by \citet[Theorem 2]{eisenberg2001systemic}, we recover a unique clearing payment vector $p(x)$ for every $x \in \R^N_+$ due to the strictly positive obligations to society. In the following proposition, we provide properties of the clearing payment vector.

\begin{proposition}\label{prop:EN}
Consider a financial network of $N$ banks with obligations $\bar p \in \R^{N \times (N+1)}_+$ satisfying $\bar p_{i0} > 0$ for each $i\in [N]$. For every asset level $x \in \R^N_+$, there exists a unique clearing payment vector $p(x) \in \R^N_+$ satisfying \eqref{eq:EN}. Furthermore, the function $p\colon \R^N_+ \to \R^N_+$ is non-expansive\footnote{That is, $\sum_{i\in[N]}|p_i(x)-p_i(x^\prime)|\leq \sum_{i\in[N]}|x_i-x^\prime_i|$ for every $x,x^\prime\in\R^N_+$.}, concave, increasing, submodular, and it holds $p_j(x + \delta e_i) \leq p_j(x + \delta \pi_{i\cdot})$ for every $i,j\in[N]$ with $i \neq j$, $x \in \R^N_+$, and $\delta \geq 0$.
\end{proposition}

\begin{proof}
Due to the strictly positive obligations to society, the provided system is a regular network as per \citet[Definition 5]{eisenberg2001systemic}. Therefore, by \citet[Theorem 2]{eisenberg2001systemic}, the uniqueness of the clearing payment vector $p(x)$ follows for every $x \in \R^N_+$. Furthermore, by \citet[Lemma 5]{eisenberg2001systemic}, the clearing payment function $p\colon \R^N_+\to\R^N_+$ is concave, increasing, and non-expansive; its submodularity follows from \citet[Proposition A.7]{banerjee2022pricing}.

Finally, let us fix $i\in[N]$, $x \in \R^N_+$, and $\delta \geq 0$. Note that $p_i(x + \delta \pi_{i\cdot}) + \delta \geq p_i(x) + \delta \geq p_i(x + \delta e_i)$ by a simple application of the monotonicity and non-expansiveness of the clearing payment function. Let $J \coloneqq \{j \in [N]\setminus\{i\} \; | \; p_j(x + \delta \pi_{i\cdot}) < \sum_{k = 0}^N \bar p_{jk}\}$ be the collection of banks (not $i$) that fail to make their payments in full under asset level $x+\delta\pi_{i\cdot}\in\R^N_+$. For every $j \in J^c:=([N]\setminus\{i\})\setminus J$, we have $p_j(x + \delta \pi_{i\cdot}) = \sum_{k = 0}^N \bar p_{jk} \geq p_j(x + \delta e_i)$ by construction. Assume that $J \neq \emptyset$ as otherwise the proof is completed. By the construction of clearing payments, with an abuse of notation for sub-vectors and sub-matrices, we have
\begin{align*}
p_J(x + \delta \pi_{i\cdot}) &= x_J + \pi_{iJ}^\T [p_i(x + \delta \pi_{i\cdot}) + \delta] + \pi_{J^c J}^\T \sum_{k = 0}^N \bar p_{J^c k} + \pi_{JJ}^\T p_J(x + \delta \pi_{i\cdot}) \\
&= (I_{|J|} - \pi_{JJ}^\T)^{-1} \left(x_J + \pi_{iJ}^\T [p_i(x + \delta \pi_{i\cdot}) + \delta] + \pi_{J^c J}^\T \sum_{k = 0}^N \bar p_{J^c k}\right) \\
&\geq (I_{|J|} - \pi_{JJ}^\T)^{-1} \left(x_J + \pi_{iJ}^\T p_i(x + \delta e_i) + \pi_{J^c J}^\T p_{J^c}(x + \delta e_i)\right) \geq p_J(x).
\end{align*}
In this calculation, the second line follows from rearranging terms and noting, due to the obligations to society, that $I_{|J|} - \pi_{JJ}^\T$ is invertible. The first inequality takes advantage of the positivity of the Leontief inverse as well as bounds on $p_k(x + \delta \pi_{i\cdot})$ for $k \in J^c \cup \{i\}$. The final inequality follows directly from the construction of the clearing payments.
\end{proof}

As detailed in, e.g.,~\cite{feinstein2017measures}, the payments to society form a meaningful aggregation function for measuring systemic risk. The idea is that, \textit{a priori}, a regulator does not care if a bank fails; instead, the risks are whether these failures spread to society, i.e., the real economy. Assuming that the regulator requires at least $\gamma \in (0,1)$ fraction of all obligations to society to be repaid, the corresponding \emph{Eisenberg-Noe aggregation function} $\bar\Lambda^{EN}: \R^N \to \R$ is defined by
\begin{equation*}
\bar\Lambda^{EN}(x) := \sum_{i \in [N]} \left(\pi_{i0} p_i(x) - \gamma \bar p_{i0}\right)
\end{equation*}
for every $x \in \R^N_+$, and by $\bar\Lambda^{EN}(x):=-\infty$ otherwise.  

Immediately upon the introduction of the single-element aggregation function $\bar\Lambda^{EN}$, there is a natural decomposition with financial interpretation. Specifically, rather than aggregating all payments to society together, the regulator can impose the condition that each bank $i\in[N]$ needs to cover $\gamma$ fraction of its own obligations to society, i.e., we may define
\[
\bar\Lambda^{EN}_i(x) := \pi_{i0} p_i(x) - \gamma \bar p_{i0}
\]
whenever $x \in \R^N_+$, and $\bar\Lambda^{EN}_i(x):=-\infty$ otherwise. It can trivially be verified that $(\bar\Lambda^{EN}_i)_{i\in[N]}$ is a single-element decomposition of $\bar\Lambda^{EN}$ in the sense of \eqref{eq:singledec}. We wish to note that, though this is a financially meaningful decomposition, it may be possible that a different, minimal, decomposition can be found via the construction outlined in Appendix~\ref{app:nash-decomposition} when a different ``natural'' decomposition was chosen, i.e., one that admits a remainder term.

\begin{proposition}\label{prop:EN-decomp}
Consider a financial network of $N$ banks with obligations $\bar p \in \R^{N \times (N+1)}_+$ satisfying $\bar p_{i0} > 0$ for each $i\in [N]$. The collection $(\bar\Lambda_i^{EN})_{i \in [N]}$ is a single-element decomposition of the Eisenberg-Noe aggregation function $\bar\Lambda^{EN}$ satisfying all properties in Corollary~\ref{cor:sensitive} with a self-preferential constant $L$ that satisfies
\[
L \leq \max_{\lambda\in \{0,1\}^N}\max_{\substack{i,j\in [N]\colon\\ i\neq j}} \frac{[(I_N - \diag(\lambda)\pi)^{-1}]_{ij}}{[(I_N - \diag(\lambda)\pi)^{-1}]_{jj}}\leq \max_{i\in [N]}(1-\pi_{i0})<1.
\]
\end{proposition}

\begin{proof}
First, we note that $\dom\bar\Lambda_i^{EN} = \dom\bar\Lambda^{EN} = \R^N_+$ for all banks $i\in[N]$, hence the consistent domain property holds. Additionally, we have $\sum_{i\in[N]}\bar\Lambda^{EN}_i=\bar\Lambda^{EN}$ so that $(\bar\Lambda^{EN}_i)_{i\in[N]}$ is dominated by $\bar\Lambda^{EN}$. It remains to show that the additional properties in Corollary~\ref{cor:sensitive} hold. Self-feasibility and self-infeasibility follow with $\bar x_i := \sum_{j = 0}^N \bar p_{ij}$ and $\underline x_i := 0$ for every bank $i \in [N]$ as we have, trivially, $p_i(\bar x_i , x_{-i}) = \sum_{j = 0}^N \bar p_{ij}$ and $(\underline x_i-\varepsilon , x_{-i}) \not\in \R^N_+$ for every $x \in \R^{N}_+$ and $\varepsilon > 0$.
Finally, the exact formulation of the self-preferential constant $L$ comes from the worst case using the derivatives of the clearing payments as found in, e.g., \cite{liu2010sensitivity}, i.e., \[\frac{\partial}{\partial x_j} p_i(x) = [(I_N - \diag(\ind_{\{p(x) < \sum_{k = 0}^N \bar p_{\cdot k}\}})\pi^\T)^{-1}]_{ij}\ind_{\{p_j(x) < \sum_{k = 0}^N \bar p_{jk}\}}.\]
Therefore, locally in the default region $\lambda = \ind_{\{p(x) < \sum_{k = 0}^N \bar p_{\cdot k}\}}$, 
\[L(x) = \max_{\substack{i \neq j \\ \lambda_i = 1}} \frac{[(I_N - \diag(\lambda)\pi^\T)^{-1}]_{ij}\lambda_j}{[(I_N - \diag(\lambda)\pi^\T)^{-1}]_{ii}} = \max_{\substack{i \neq j \\ \lambda_j = 1}} \frac{[(I_N - \diag(\lambda)\pi)^{-1}]_{ji}\lambda_i}{[(I_N - \diag(\lambda)\pi)^{-1}]_{jj}}\]
where the equality follows from the use of the Leontief inverses $(I_N - \diag(\lambda)\pi^\T)^{-1}$ and $(I_N - \diag(\lambda)\pi)^{-1}$.
Immediately, the first inequality becomes clear:
\[L \leq \max_{\lambda\in \{0,1\}^N}\max_{\substack{i,j\in [N]\colon\\ i\neq j}} \frac{[(I_N - \diag(\lambda)\pi)^{-1}]_{ij}}{[(I_N - \diag(\lambda)\pi)^{-1}]_{jj}}.\]  
Next, we bound this constant by $\max_{i \in [N]} (1-\pi_{i0})$.

Let us fix a binary vector $\lambda \in \{0,1\}^N$ and define $E(\lambda):=\{i\in [N]\; | \; \lambda_i=0\}$ and $E(\lambda)^c:=[N]\setminus E(\lambda)$. Note that $\diag(\lambda)\pi$ is the matrix obtained by replacing the $i^{\text{th}}$ row of $\pi$ with a row of zeroes for each $i\in E(\lambda)$. Let $\tilde{\pi}(\lambda)=(\pi_{ij}(\lambda))_{i,j\in[N]\cup\{0\}}\in\R^{(N+1)\times (N+1)}_+$ be the relative liabilities matrix of the extended network in which node $0$ represents society and nodes $1,\ldots,N$ represent the banks. Since society does not have any obligations to the banks, we take $\pi_{00}(\lambda):=1$ and $\pi_{0j}(\lambda):=0$ for each $j\in [N]$. Moreover, in this network, we assume that all banks in $E(\lambda)$ have their full obligations to society, i.e.,
\[
\pi_{ij}(\lambda)=\begin{cases}1 & \text{if } i\in E(\lambda)\cup\{0\},\ j=0,\\ 0 &\text{if }i\in E(\lambda)\cup\{0\},\ j\in [N],\\
\pi_{ij}& \text{if }i\in E(\lambda)^c,\ j\in [N]\cup\{0\}.\end{cases}
\]
Note that $\tilde{\pi}(\lambda)$ is a right stochastic matrix and we may view it as the one-step transition probability matrix of a time-homogeneous Markov chain with state space $[N]\cup\{0\}$. Clearly, state $0$ is the only recurrent state and all other states are transient. Moreover, since state $0$ is absorbing, given $i,j\in [N]$, there is no path from $i$ to $j$ that goes through state $0$; in particular, the corresponding multi-step transition probabilities can be calculated by only using the original matrix $\diag(\lambda)\pi$ via
\[
[(\diag(\lambda)\pi)^{n}]_{ij}=[(\tilde{\pi}(\lambda))^{n}]_{ij},\quad n\in\mathbb{Z}_+.
\]
Moreover, we have
\[
(I_N-\diag(\lambda)\pi)^{-1}=\sum_{n=0}^\infty (\diag(\lambda)\pi)^n,\quad (I_{N+1}-\tilde{\pi}(\lambda))^{-1}=\sum_{n=0}^\infty (\tilde{\pi}(\lambda))^n.
\]
Hence, for every $i,j\in [N]$, we also have
\[
[(I_N-\diag(\lambda)\pi)^{-1}]_{ij}=[(I_{N+1}-(\tilde{\pi}(\lambda))^{-1}]_{ij}.
\]

The matrix $a(\lambda)=(a_{ij}(\lambda))_{i,j\in [N]\cup\{0\}}:=(I_{N+1}-\tilde{\pi}(\lambda))^{-1}$ has a well-known interpretation in the theory of Markov chains. For each $j\in [N]\cup\{0\}$, let $N_j(\lambda)$ be the total number of visits of the chain to state $j$. Then, for each $i\in [N]\cup\{0\}$, we have $a_{ij}(\lambda)=\E_i [N_j(\lambda)]$, where $\E_i$ denotes the conditional expectation given that the chain starts at state $i$ at time zero. Since state $0$ is the unique absorbing state, we have $a_{0j}(\lambda)=0$ for every $j\in [N]$ and $a_{i0}(\lambda)=+\infty$ for every $i\in [N]\cup\{0\}$.

For every $i,j\in [N]$ with $i\neq j$, by \citet[Corollary~5.2.13]{cinlar1975}, we have
\[
a_{ij}(\lambda)=\mathbb{P}_i\{T_j(\lambda)<+\infty\}a_{jj}(\lambda),
\]
where $T_j(\lambda)$ is the first time that the chain visits state $j$ after time zero and $\mathbb{P}_i$ denotes the conditional probability given that the chain starts at state $i$ at time zero. Note that
\[
\mathbb{P}_i\{T_j(\lambda)=+\infty\}\geq \pi_{i0}(\lambda)=\begin{cases}1& \text{if } i\in E(\lambda),\\ \pi_{i0} & \text{if } i\in E(\lambda)^c,\end{cases}
\]
since state $0$ is absorbing. Hence, $\mathbb{P}_i\{T_j(\lambda)<+\infty\}=0$ if $i\in E(\lambda)$ and $\mathbb{P}_i\{T_j(\lambda)<+\infty\}\leq 1-\pi_{i0}$ if $i\in E(\lambda)^c$. Then, we get
\begin{align*}
\max_{\substack{i,j\in [N]\colon\\ i\neq j}}\frac{[(I_N-\diag(\lambda)\pi)^{-1}]_{ij}}{[(I_N-\diag(\lambda)\pi)^{-1}]_{jj}}&=\max_{\substack{i,j\in [N]\colon\\ i\neq j}}\frac{[(I_{N+1}-\tilde{\pi}(\lambda))^{-1}]_{ij}}{[(I_{N+1}-\tilde{\pi}(\lambda))^{-1}]_{jj}}=\max_{\substack{i,j\in [N]\colon\\ i\neq j}}\frac{a_{ij}(\lambda)}{a_{jj}(\lambda)}\\
&\leq \max_{\substack{i,j\in [N]\colon\\ i\neq j}}\mathbb{P}_i\{T_j(\lambda)<+\infty\}\leq \max_{i\in E(\lambda)^c}(1-\pi_{i0}).
\end{align*}
Therefore,
\[
L\leq\max_{\lambda\in \{0,1\}^N}\max_{\substack{i,j\in [N]\colon\\ i\neq j}}\frac{[(I_N-\diag(\lambda)\pi)^{-1}]_{ij}}{[(I_N-\diag(\lambda)\pi)^{-1}]_{jj}}
\leq \max_{\lambda \in \{0,1\}^N} \max_{i\in E(\lambda)^c}(1-\pi_{i0})=\max_{i\in E\setminus\{0\}} (1-\pi_{i0})<1
\]
since we assume that $\pi_{i0}>0$ for every $i\in [N]$.
\end{proof}

\subsection{Nash Allocation}\label{sec:EN-nash}

Given the Eisenberg-Noe aggregation function $\bar\Lambda^{EN}$ and its natural decomposition $(\bar\Lambda_i^{EN})_{i \in [N]}$, we wish to consider various properties of the associated Nash allocations for sensitive systemic risk measures. We note that the Nash allocations for the insensitive systemic risk measures can trivially be found (Corollary~\ref{cor:insensitive}).

As we only consider the sensitive systemic risk measures within this section, we define 
\[
\Lambda^{EN}(x,m) = \bar\Lambda^{EN}(x+m),\quad \Lambda_i^{EN}(x,m) = \bar\Lambda_i^{EN}(x+m)
\]
for each $x,m\in\R^N$ and $i\in[N]$ throughout. We also fix an acceptance set $\A$ satisfying Assumption~\ref{asmp:coherent}. First, within Corollary~\ref{cor:EN}, we summarize the existence of a Nash allocation. We then consider a single optimization problem which can be used to find a Nash allocation for the natural decomposition. Finally, we conclude with a quick discussion on two situations in which the Nash allocation is unique.

\begin{corollary}\label{cor:EN}
Consider a financial network of $N$ banks with obligations $\bar p \in \R^{N \times (N+1)}_+$ satisfying $\bar p_{i0} > 0$ for each $i\in [N]$.
\begin{enumerate}
\item There exists a $(\A,(\Lambda_i^{EN})_{i \in [N]})$-Nash allocation rule $r^*\colon L^\infty(\R^N)\to\R^N$.
\item Suppose that $\A$ is the acceptance set of a coherent optimized certainty equivalent risk measure as in Corollary~\ref{cor:comonotonic}. Let $X\in L^\infty(\R^N)$ be a stress scenario and denote by $Z^X$ the comonotonic copula of $X$. Then, $r^*(Z^X) \in \R^N$ is acceptable under $X$, i.e., $\bar\Lambda^{EN}(X+r^*(Z^X)) \in \A$.
\end{enumerate}
\end{corollary}
\begin{proof}
The existence of a $(\A,(\Lambda_i^{EN})_{i \in [N]})$-Nash allocation rule follows from Corollary~\ref{cor:sensitive} as the natural decomposition satisfies all properties of that result (Proposition~\ref{prop:EN-decomp}). When $\A$ is the acceptance set of a coherent optimized certainty equivalent risk measure as in Corollary~\ref{cor:comonotonic}, by the submodularity of the Eisenberg-Noe aggregation function (Proposition~\ref{prop:EN}), the acceptability of the comonotonic Nash allocation $r^*(Z^X)$ follows from Corollary~\ref{cor:comonotonic}.
\end{proof}

Though we only guarantee the existence of a Nash allocation in Corollary~\ref{cor:EN}, we are able to provide a constructive approach to compute one such allocation. In particular, the computed allocation will be the one that has the minimal total capital allocation. To construct this allocation, we are inspired by \citet[Proposition 2]{braouezec2023economic}; in contrast to that result, we are not considering the shared constraint game; however, herein, the same approach can be used to find a $(\A,(\Lambda_i^{EN})_{i \in [N]})$-Nash allocation with its \emph{individual} constraints.
Notably, under a finite probability space with coherent (polyhedral) acceptance set, the approach provided in Theorem~\ref{thm:optimization} below only requires solving a single convex (resp.\ linear) program.

\begin{theorem}\label{thm:optimization}
Consider a financial network of $N$ banks with obligations $\bar p \in \R^{N \times (N+1)}_+$ satisfying $\bar p_{i0} > 0$ for each $i\in [N]$. For each $X \in L^\infty(\R^N)$, there exists a solution to the optimization problem
\begin{equation}\label{eq:opt}
R^*(X) := \argmin_{m \in \R^N}\left\{\sum_{i \in [N]} m_i \; | \; \bar\Lambda_i^{EN}(X+m) \in \A \, \forall i \in [N]\right\} \neq \emptyset.
\end{equation}
Moreover, if $r^*\colon L^\infty(\R^N) \to \R^N$ is a selector of $R^*$, i.e., $r^*(X) \in R^*(X)$ for every $X \in L^\infty(\R^N)$, then $r^*$ is a $(\A,(\Lambda_i^{EN})_{i \in [N]})$-Nash allocation rule.
\end{theorem}

\begin{proof}
Let $X\in L^\infty(\R^N)$. We first argue that $R^*(X) \neq \emptyset$. By Theorem~\ref{thm:exist-unique}, there exists a Nash allocation rule $r\colon L^\infty(\R^N)\to\R^N$. In particular, $\bar\Lambda^{EN}_i(X+r(X))\in\A$ for each $i\in[N]$. Consequently, the feasible region of the problem in \eqref{eq:opt} is a non-empty set. Moreover, it is also a closed set, which can be checked similar to the proof of Corollary~\ref{cor:acceptable}. Next, for a feasible solution $m\in\R^N$, the acceptance set constraints implicitly imply that $X+m\in \dom\bar\Lambda^{EN}=\R^N_+$, i.e., $m_i\geq -\essinf X_i$ for every $i\in[N]$. Let $v(X)$ denote the optimal value of the problem. Then, we have $v(X)\in\R$ and
\[
v(X)=\negthinspace\inf_{m\in\R^N}\negthinspace\cb{\sum_{i\in [N]}m_i\; | \; \bar\Lambda^{EN}_i(X+m)\in \A\ \forall i\in[N], -\essinf X\leq m, \sum_{i\in[N]}m_i\leq \sum_{i\in[N]}\bar{m}_i(X)},
\]
where $\bar{m}(X)\in\R^N$ is an arbitrarily fixed feasible solution of the original problem. Since the feasible region of the new problem is compact and the objective function is continuous, there exists an optimal solution $r^*(X)$, which is also an optimal solution for the original problem.

Next, we show that every selector $r^*\colon L^\infty(\R^N)\to \R^N$ of $R^*$ is a Nash allocation rule. Note that a vector $m \in \R^N$ satisfies the definition of a Nash allocation rule at $X \in L^\infty(\R^N)$ if and only if, for every $i\in[N]$, we have $\rho(\bar\Lambda^{EN}_i(X+m)) \leq 0$ and the following implication holds:
\[
\rho(\bar\Lambda^{EN}_i(X+m)) < 0 \quad \Rightarrow \quad m_i = -\essinf X_i.
\]
To get a contradiction, suppose that $r^*(X)$ does not satisfy this definition. Then, using the notation of $\phi^X$ from the proof of Theorem~\ref{thm:exist-unique}, there exists some $i\in [N]$ such that $m^*_i(X) := \phi^X_i(r^*(X)) < r_i^*(X)$. Let $\delta \in (0,r_i^*(X) - m^*_i(X)]$.
Therefore, for every $j\in [N]\setminus\{i\}$, we have
\[
\rho(\bar\Lambda_j^{EN}(X + r^*(X) + \delta\pi_{i\cdot} - \delta e_i)) \leq \rho(\bar\Lambda_j^{EN}(X+r^*(X))) \leq 0
\] 
by the final property in Proposition~\ref{prop:EN}.
Additionally, by the monotonicity of the Eisenberg-Noe decomposition, the choice of $\delta$, and the definition of $m^*_i(X)$, we have
\[
\rho(\bar\Lambda_i^{EN}(X + r^*(X) + \delta(\pi_{i\cdot} - e_i))) \leq \rho(\bar\Lambda_i^{EN}(X + r^*(X) - \delta e_i)) \leq \rho(\bar\Lambda_i^{EN}(X + (\bar m^*_i(X),r_{-i}^*(X))) \leq 0.\]
Therefore, $r^*(X) + \delta(\pi_{i\cdot} - e_i)$ is feasible for the problem in \eqref{eq:opt} and we have $\sum_{j\in[N]} (r_j^*(X) + \delta(\pi_{ij} - e_{ij}))=\sum_{j\in[N]} r_j^*(X) - \pi_{i0} < \sum_{j\in[N]} r_j^*(X)$ since $\pi_{i0}>0$, contradicting the optimality of $r^*(X)$.
\end{proof}

Thus far, we have considered the existence of Nash allocation rules (Corollary~\ref{cor:EN}) and formulated an optimization problem to compute one such rule (Theorem~\ref{thm:optimization}). We conclude our general discussion on the Nash allocation rules for the sensitive Eisenberg-Noe systemic risk measure by providing two cases in which uniqueness of the Nash allocation can be guaranteed.

\begin{corollary}\label{cor:EN-unique}
Consider a financial network of $N$ banks with obligations $\bar p \in \R^{N \times (N+1)}_+$ satisfying $\bar p_{i0} > 0$ for each $i\in [N]$.
\begin{enumerate}
\item If $\pi_{i0} > \frac{N-2}{N-1}$ for every bank $i\in[N]$, then there exists a unique $(\A,(\Lambda_i^{EN})_{i \in [N]})$-Nash allocation rule.
\item Consider a deterministic shock $X = x$, where $x \in \R^N$. Then, every $(\A,(\Lambda_i^{EN})_{i \in [N]})$-Nash allocation rule takes the same value at $X$.
\end{enumerate}
\end{corollary}

\begin{proof}
The first claim is a trivial consequence of Theorem~\ref{thm:exist-unique:2} using the bound on the self-preferential constant $L \leq \max_{i\in[N]} (1 - \pi_{i0})$ from Proposition~\ref{prop:EN-decomp}.

To prove the second statement, let $X=x$ with $x\in\R^N$. Let $i\in[N]$. By construction, $\rho(\bar\Lambda_i^{EN}(x+m)) = -\bar\Lambda_i^{EN}(x+m)$ for every $m \geq -x$. Therefore, the acceptability constraint for bank $i$ is equivalent to $p_i(x+m) \geq \gamma \sum_{j = 0}^N \bar p_{ij}$. We take advantage of this structure in order to provide the following algorithm to compute a Nash allocation $m(x)\in\R^N$:
\begin{enumerate}
\item Initialize $k \leftarrow 0$ and set $I^0 \coloneqq \emptyset$.
\item\label{step2} Define $p_i^k \coloneqq \gamma \sum_{j = 0}^N \bar p_{ij}$ if $i \not\in I^k$. Define $(p_i^k)_{i \in I^k}$ as the unique solution of the Eisenberg-Noe clearing problem
    \[
    p_i^k = \min\left\{\sum_{j = 0}^N \bar p_{ij} \; , \; \gamma \sum_{j \not\in I^k} \bar p_{ji} + \sum_{j \in I^k} \pi_{ji} p_j^k\right\} \quad \forall i \in I^k.
    \]
\item Define $m^k$ by
    \[
    m_i^k \coloneqq \begin{cases} -x_i + \gamma \sum_{j = 0}^N \bar p_{ij} - \sum_{j = 1}^N \pi_{ji} p_j^k &\text{if } i \not\in I^k \\ -x_i &\text{if } i \in I^k. \end{cases}
    \]
\item Define $I^{k+1} \coloneqq \{i \in [N] \; | \; m_i^k \leq -x_i\}$ and iterate $k \leftarrow k+1$.
\item If $I^k = I^{k-1}$, then terminate the algorithm with $m(x) \coloneqq m^{k-1}$ and $I^* \coloneqq I^{k-1}$. Otherwise, return to step~\eqref{step2}.
\end{enumerate}
Note that this algorithm converges in at most $N$ iterations as $I^k \supseteq I^{k-1}$ for every $k \geq 1$.
Furthermore, by construction, it is clear that $m(x)$ is a Nash allocation, i.e., $m_i(x) = \inf\{m_i\in\R\mid \bar\Lambda^{EN}_i(x+(m_i,m_{-i}(x)))\in\A\}$ for every $i\in[N]$, as no bank can reduce their capital requirement without violating the acceptability constraint.

Finally, let $r\colon L^\infty(\R^N)\to\R^N$ be a Nash allocation rule. As in the above algorithm, $r(x)$ can be associated with the set of banks with no capital requirement $I := \{i \in [N] \; | \; r_i(x) \leq -x_i\}$. By the construction of $I^*$ to be the smallest possible set of banks who can passively maintain acceptability, it follows that $I \supseteq I^*$. Assume $I \supsetneq I^*$ and fix some $i \in I \backslash I^*$, i.e., $r_i(x) = -x_i < m_i(x)$. However, by simple inspection, $p_i(x+r) < \gamma \sum_{j = 0}^N \bar p_{ij}$ (or else $i \in I^*$ must hold). Therefore $I = I^*$ and, as a consequence, $r(x) = m(x)$, completing the proof.
\end{proof}

\begin{remark}
Though Corollary~\ref{cor:EN-unique} provides sufficient conditions for the uniqueness of a Nash allocation, the authors have not been able to construct any examples with multiple Nash allocations for the Eisenberg-Noe aggregation function. This hints that a stronger result may be possible, but we leave that for future work.
\end{remark}

\subsection{Numerical Examples}\label{sec:EN-numeric}

We conclude our discussion of the Eisenberg-Noe aggregation function by presenting two detailed case studies where we compute explicitly the Nash allocations under given stress scenarios. In both examples, we choose settings from Corollary~\ref{cor:EN-unique} that guarantee the uniqueness of the Nash allocations: first, in Section~\ref{sec:EN-numeric-2}, we consider the $N = 2$-bank setting; second, in Section~\ref{sec:EN-numeric-EBA}, we consider a deterministic setting calibrated to data from the European Banking Authority (EBA). Throughout, we compare the Nash allocations with the minimal possible systemically acceptable capital requirements. In all cases, though the Nash allocations may require strictly larger total capital, the Nash allocations provide financially justifiable selections from the set of all acceptable capitals.

\subsubsection{Two-Bank Example}\label{sec:EN-numeric-2}

Consider the $N = 2$-bank system with obligations $\bar p_{10}=\bar p_{20} = \bar p_{12} = 1$, $\bar p_{21} = \frac12$, and $\bar p_{11}=\bar p_{22}=0$. Furthermore, we set acceptability such that $\gamma = 95\%$ of obligations to society are recovered.
Recall from Corollary~\ref{cor:EN-unique} that this network immediately admits a unique Nash allocation for any stress scenario and coherent risk measure.
For the numerics, we set $X_i \sim U_i \sum_{j = 0}^N \bar p_{ij}$ for each bank $i\in[N]$, where $U=(U_1,\ldots,U_N)^{\T}$ is a vector of standard uniform random variables generated from a Gaussian copula with correlations $\frac12$.\footnote{The stress scenario is implemented with $1000$ sample values.}
Finally, we will consider two risk measures within this example:
\begin{enumerate}
\item \emph{Expectation risk measure}: $\A^\E := \{Y \in L^\infty(\R) \; | \; \E[Y] \geq 0\}$; and 
\item \emph{Average Value-at-Risk}: 
$\A^{AVaR} := \{Y \in L^\infty(\R) \; | \; \frac{1}{\alpha}\int_0^\alpha VaR_\gamma(Y) d\gamma \leq 0\}$ at $\alpha = 5\%$, where $VaR_{\gamma}(Z) := \inf\{m \in \R \; | \; \Pr\{X+m < 0\} \leq \gamma\}$.
\end{enumerate}

\begin{table}[ht]
\begin{center}
\begin{tabular}{|r|c|c|c||c|c|c|}
\multicolumn{1}{r}{} & \multicolumn{6}{c}{\textbf{Insensitive Systemic Risk Measure}} \\ \cline{2-7}
\multicolumn{1}{r|}{} &\multicolumn{3}{c||}{\textbf{Expectation Risk Measure}} & \multicolumn{3}{c|}{\textbf{Average Value-at-Risk}} \\ \cline{2-7}
\multicolumn{1}{r|}{} & \textbf{Minimal} & \textbf{Nash} & \textbf{Comonotonic} & \textbf{Minimal} & \textbf{Nash} & \textbf{Comonotonic} \\ \hline
\textbf{Bank 1} & \cellcolor{lightgray} & 0.2772 & 0.2865 & \cellcolor{lightgray} & 0.8723 & 0.9088 \\ \hline
\textbf{Bank 2} & \cellcolor{lightgray} & 0.1657 & 0.1996 & \cellcolor{lightgray} & 0.8299 & 0.8946 \\ \hline\hline
\textbf{Total}  & 0.4429 & 0.4429 & 0.4861 & 1.6867 & 1.7022 & 1.8034 \\ \hline
\multicolumn{7}{c}{} \\
\multicolumn{1}{r}{} & \multicolumn{6}{c}{\textbf{Sensitive Systemic Risk Measure}} \\ \cline{2-7}
\multicolumn{1}{r|}{} &\multicolumn{3}{c||}{\textbf{Expectation Risk Measure}} & \multicolumn{3}{c|}{\textbf{Average Value-at-Risk}} \\ \cline{2-7}
\multicolumn{1}{r|}{} & \textbf{Minimal} & \textbf{Nash} & \textbf{Comonotonic} & \textbf{Minimal} & \textbf{Nash} & \textbf{Comonotonic} \\ \hline
\textbf{Bank 1} & 0.8994 & 0.9201 & 0.9403 & 1.3311 & 1.3569 & 1.3703 \\ \hline
\textbf{Bank 2} & 0.1697 & 0.1502 & 0.2192 & 0.3963 & 0.4024 & 0.4331 \\ \hline\hline
\textbf{Total}  & 1.0691 & 1.0703 & 1.1595 & 1.7275 & 1.7593 & 1.8034 \\ \hline
\end{tabular}
\end{center}
\caption{Section~\ref{sec:EN-numeric-2}: Summary table of capital requirements under different systemic risk measures and utilizing (i) minimal capital requirements; (ii) the Nash allocation rule; and (iii) the Nash allocation rule under the comonotonic copula.}
\label{table:2bank}
\end{table}

Within Table~\ref{table:2bank}, we can see how the choice of aggregation and acceptance set alter the capital requirements assessed. Throughout, we take the minimal capital requirements as the minimizer of $\{m_1 + m_2 \; | \; m \in R(X)\}$ for the appropriate choice of systemic risk measure $R$. Notably, for the insensitive case, because this systemic risk measure is a (truncated) half-space, we do not record the individual minimal requirements but only the total capital requirement assessed. In addition, we benchmark these results by displaying the Nash allocations when only the marginal distributions for each bank are known (i.e., under the comonotonic copula). In every setting, we notice that the total minimal requirements are lower than the aggregate Nash allocations, which in turn are lower than the comonotonic copula case. In fact, the Nash allocations for each bank are lower than their counterpart under the comonotonic copula; this, however, does not hold between the minimal capital requirements and the Nash allocations. The sensitive systemic risk measures and capital allocations are plotted in Figure~\ref{fig:2bank}.

\begin{figure}[ht]
\centering
\begin{subfigure}[t]{0.45\textwidth}
\includegraphics[width=\textwidth]{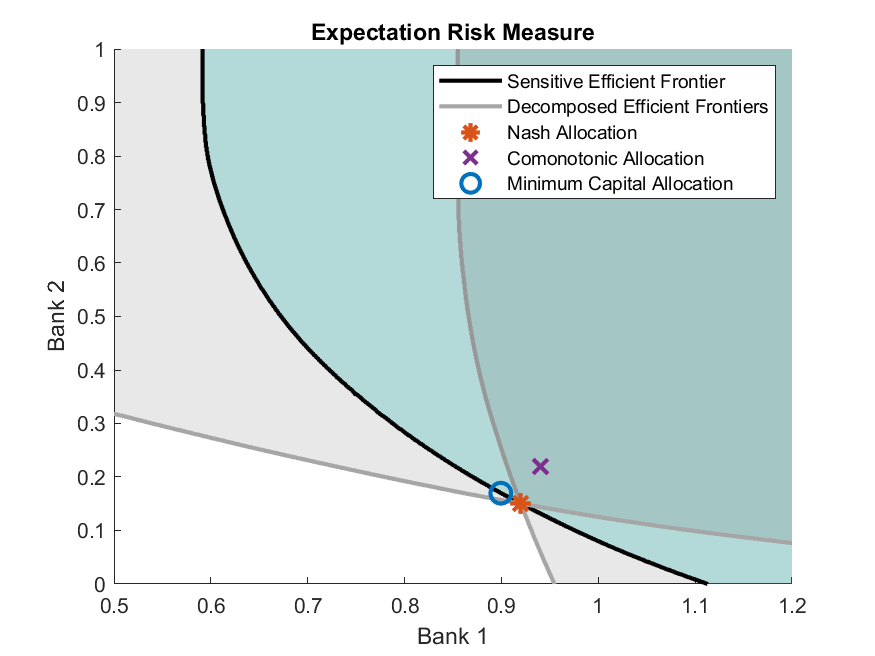}
\caption{Expectation risk measure}
\label{fig:2bank-expectation}
\end{subfigure}
~
\begin{subfigure}[t]{0.45\textwidth}
\includegraphics[width=\textwidth]{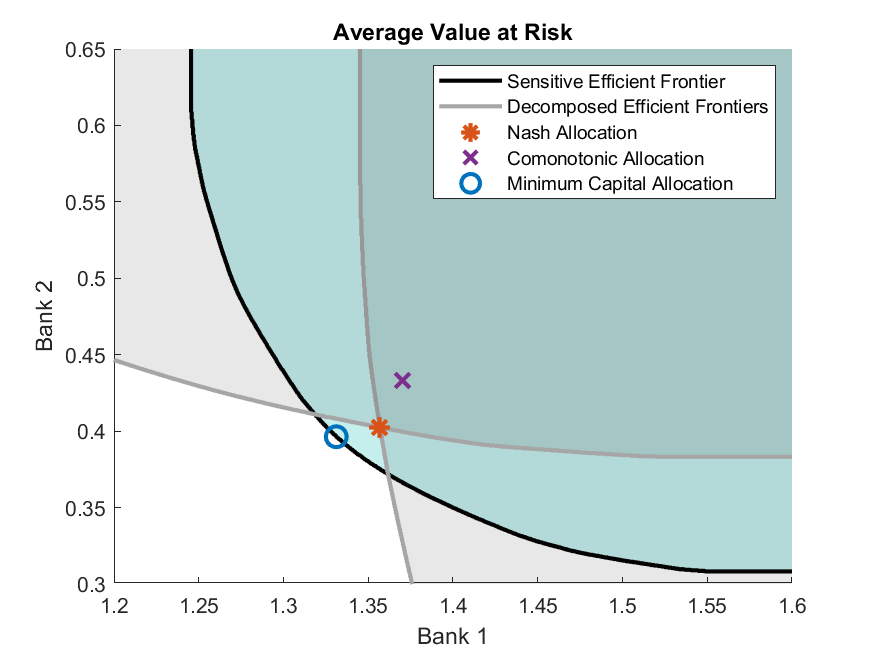}
\caption{Average Value-at-Risk}
\label{fig:2bank-avar}
\end{subfigure}
\caption{Section~\ref{sec:EN-numeric-2}: Visualizations of the sensitive systemic risk measures
and different allocation rules under different risk measures.}
\label{fig:2bank}
\end{figure}

\subsubsection{European Banking Data}\label{sec:EN-numeric-EBA}
Herein, we consider a larger system with $N = 87$ banks. We calibrate this system to the 2011 EBA stress testing dataset as also undertaken in, e.g.,~\cite{gandy2017bayesian,chen2016optimization,feinstein2018sensitivity}.\footnote{The exact calibration methodology is provided in Appendix~\ref{sec:calibration}.} In comparison to Section~\ref{sec:EN-numeric-2} above, in this example, we will consider a deterministic shock in which the initial external assets $X_i$ of bank $i\in[N]$ are as reported in the EBA dataset less 10\% of its total exposures.\footnote{Note that, as there is no randomness, the comonotonic copula of $X$ is itself.} 
Due to the deterministic shock, every coherent risk measure provides the same results; therefore, we can completely characterize the set of acceptable allocations by the threshold $\gamma = 97.5\%$ for the total repayment to society. For the purposes of this example, we will solely consider the sensitive systemic risk measure setting. Recall from Corollary~\ref{cor:EN-unique} that this system admits a unique Nash allocation. 

Figure~\ref{fig:EBA} displays the capital requirements assigned to each bank in the EBA dataset under the sensitive systemic risk measure setting. First, it is notable how the minimal capital requirements and the Nash allocation closely coincide for many banks in the system (Figure~\ref{fig:EBA-comp-zoom}). However, as observed in Figure~\ref{fig:EBA-comp}, under the minimal capital requirements, there are a few institutions that ``freeload'' by having large negative capital requirements which need to be offset by other banks. As observed in Figure~\ref{fig:EBA-comp}, some of the minimal capital allocations are of a larger order of magnitude than the Nash allocations. In comparison, the Nash allocations appear to be fairer in the sense that the requirements for all banks are closely aligned (but do differ due to bank-specific attributes).
When we look at the total capital requirements, we find that the Nash allocations are less than \euro4.7 Million larger than the minimal acceptable capital allocation; the minimal total capital required is \euro728,316.4122 Million whereas the total Nash allocation is \euro728,321.1119 Million. Notably, this \euro4.7 Million is distributed throughout the $87$ institutions as only 4 banks have a lower capital requirement under the minimal capital allocation than under Nash.

\begin{figure}[ht]
\centering
\begin{subfigure}[t]{0.45\textwidth}
\includegraphics[width=\textwidth]{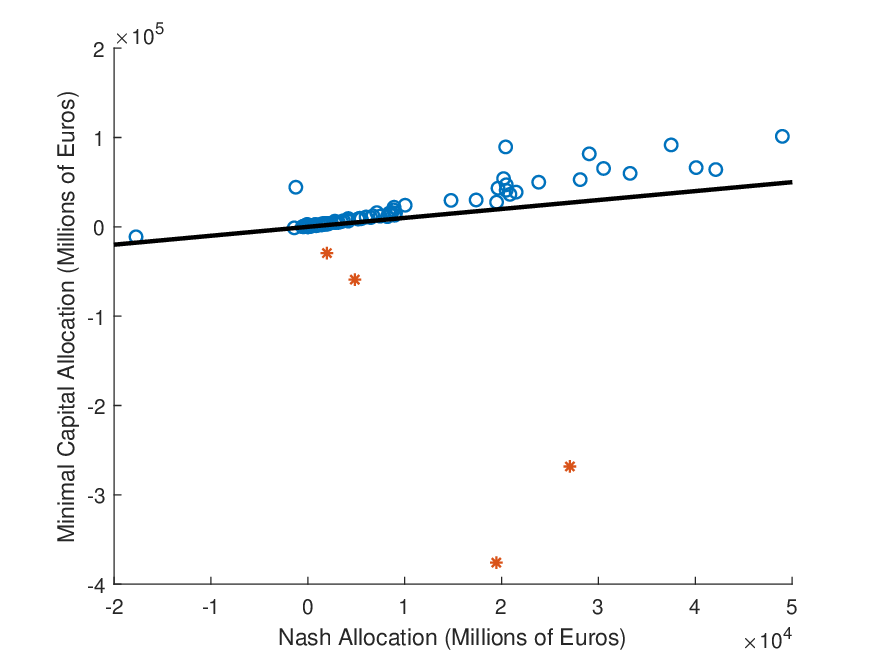}
\caption{All capital allocations.}
\label{fig:EBA-comp}
\end{subfigure}
~
\begin{subfigure}[t]{0.45\textwidth}
\includegraphics[width=\textwidth]{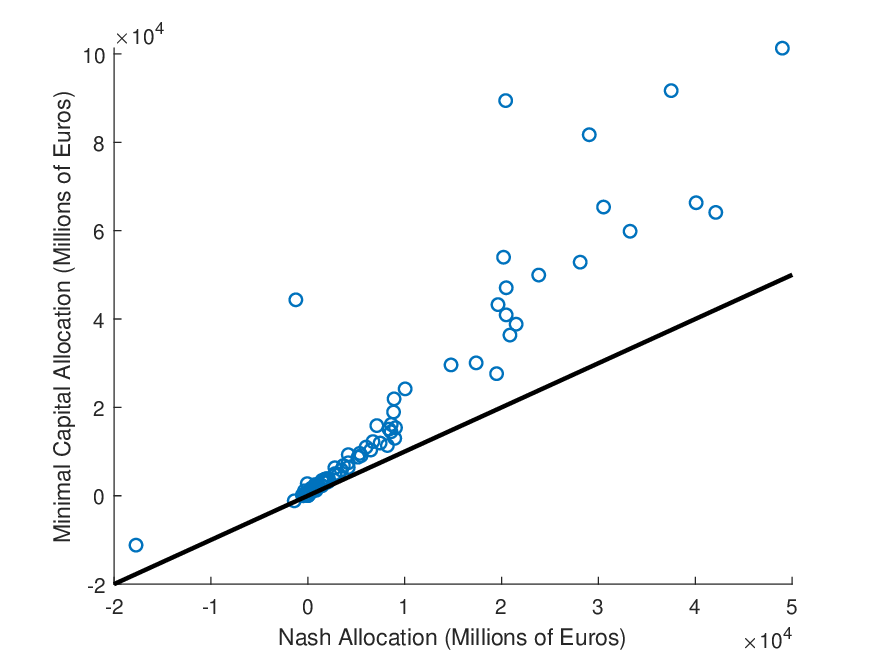}
\caption{Only banks with smaller capital requirements under the Nash allocation.}
\label{fig:EBA-comp-zoom}
\end{subfigure}
\caption{Section~\ref{sec:EN-numeric-EBA}: Visualization of the Nash (x-axis) and minimal capital allocations (y-axis). The solid black line indicates the level where the Nash and minimal capital allocations coincide.  Banks with larger capital requirements under the minimal capital allocations are blue circles; banks with larger capital requirements under the Nash allocations are red stars.}
\label{fig:EBA}
\end{figure}

\section{Conclusion}\label{sec:conclusion}
Within this work, we propose a novel and financially meaningful allocation mechanism for capital requirements. This approach, based on game theory in which institutions compete for lower capital requirements, allows regulators to assign capital requirements to banks while enforcing macro-prudential requirements. We find that the insensitive risk measures of \cite{chen2013axiomatic,kromer2016systemic} generally permit a unique Nash allocation rule, whereas the sensitive risk measures of \cite{feinstein2017measures} may require additional regularity conditions for uniqueness. We investigate the Eisenberg-Noe clearing system~\citep{eisenberg2001systemic} for systemic risk measures and numerically illustrate the associated Nash allocations for the Eisenberg-Noe aggregation function.

In Remark~\ref{rem:sd}, we argue that the scenario-dependent systemic risk measures studied in \cite{biagini2019unified,biagini2020fairness,biagini2021systemic} can be seen as a special case of sensitive risk measures with a single-element aggregation function calculated by an optimization problem. Besides this connection, we leave considerations about decomposing the scenario-dependent aggregation function $\bar\Lambda^{SD}$ for future work. We conjecture that such decompositions are best studied in light of cooperative game theory due to the inherent pooling of capital implied by the aggregation structure.

While the Eisenberg-Noe clearing system has a natural decomposition based on the clearing payments of banks, with a clear financial interpretation, choosing a suitable decomposition for the aggregation function can be a separate task in general. To address this, in Appendix~\ref{app:nash-decomposition}, we formulate an infinite-dimensional optimization problem that searches for a decomposition with the least possible self-preferential constant. We guarantee the existence of an optimal solution for this problem and illustrate a simple example where this solution can be calculated explicitly. We leave it for future research to provide a general numerical methodology to approximately solve this problem; in Remark~\ref{rem:neuralnet}, we present some first thoughts to this effect using neural network approximations.

While we focus on the existence and uniqueness of Nash allocation rules within this work, other mathematical and financial properties of these rules need to be investigated. For instance, we wish to consider if any of (lower semi-)continuity, monotonicity, translativity, or convexity can be recovered for a Nash allocation rule. We highlight that, as per~\cite{ararat2024separability}, at least one of these desirable properties cannot hold generally as interactions between firms is guaranteed by the construction.

\section*{Acknowledgments}
The first author is partially supported by T\"{U}B{\.I}TAK (Scientific \& Technological Research Council of Turkey), Project No. 123F357.

\bibliographystyle{plainnat}
\bibliography{biblio.bib}

\appendix

\section{Nash Allocations with Rescue Funds}\label{sec:rescue}

Within the main body of this text, we define Nash allocations so that each firm holds onto its own capital. However, as shown in, e.g.,~\cite{rogers2013failure}, there are cases in which rescuing another bank may lead to better outcomes for the rescuer as well as the rescued firm. With that in mind, we propose the following variation on Definition~\ref{defn:nash-allocation}.

\begin{definition}\label{defn:nash-rescue}
Let $(\Lambda_i)_{i \in [N]}$ be a decomposition of $\Lambda$. A matrix-valued functional $r\colon L^\infty_\Lambda(\R^N) \to \R^{N \times N}$ is called a \textbf{$(\A,(\Lambda_i)_{i \in [N]})$-Nash allocation rule with rescue funds} if, for every $X \in L^\infty_\Lambda(\R^N)$ and $i\in[N]$, it holds
\[
r_{i\cdot}(X) \in \argmin_{m_i \in \R^N} \left\{\sum_{j \in [N]} m_{ij} \; \left| \; \begin{array}{l} \Lambda_i(X,m_i + \sum_{j \neq i} r_{j\cdot}(X)) \in \A, \\ m_{ii}e_i + \sum_{j \neq i} r_{jj}(X)e_j \in D(X), \; m_{ij} \geq 0 \, \forall j \neq i \end{array}\right.\right\}.
\]
In this case, $r$ is called \textbf{acceptable} if $\sum_{i \in [N]} r_{i\cdot}(X) \in R(X)$ for every $X \in L^\infty_\Lambda(\R^N)$.
\end{definition}

A Nash allocation rule with rescue fund generalizes the Nash allocation rules introduced within Definition~\ref{defn:nash-allocation}. Specifically, this variation allows for firm $i$ to provide a rescue $m_{ij} \geq 0$ to firm $j$. Notably, these rescues must be nonnegative to restrict banks from stealing funds from other firms in the system.\footnote{An additional feasibility condition is imposed which guarantees the finiteness of the aggregation function for each bank when neglecting the cash transfers. This condition is imposed rather than, for bank $i$, $m_i + \sum_{j \neq i} r_{j\cdot}(X) \in D(X)$ to preclude bank $i$ from using the provided transfers $r_{ji}(X)$ to guarantee that the system is feasible, i.e., each bank guarantees that the aggregation is well-defined even if there is a delay in the transfer payments.} Overall, each institution is still attempting to minimize the total capital it is injecting into the system.

We conclude this discussion of the Nash allocation rules with rescue funds by considering when this generalization reduces to Definition~\ref{defn:nash-allocation}. Specifically, we find that the self-preferential property guarantees that any Nash allocation rule with rescue funds will not include any transfer payments.
\begin{proposition}\label{prop:self-preferential}
Let $(\Lambda_i)_{i \in [N]}$ be a \emph{self-preferential} decomposition of $\Lambda$. Let $r\colon L^\infty_\Lambda(\R^N) \to \R^{N\times N}$ define a $(\A,(\Lambda_i)_{i \in [N]})$-Nash allocation rule with rescue funds. Then, for every $X \in L^\infty_\Lambda(\R^N)$ and $i,j\in[N]$ with $i\neq j$, it holds $r_{ij}(X) = 0$. In particular, $(r_{11},\ldots,r_{NN})^{\mathsf{T}}$ is a $(\A,(\Lambda_i)_{i \in [N]})$-Nash allocation rule as defined in Definition~\ref{defn:nash-allocation}.
\end{proposition}
\begin{proof}
Let $L < 1$ be the self-preferential constant associated with $(\Lambda_i)_{i \in [N]}$. Let us fix a stress scenario $X \in L^\infty_\Lambda(\R^N)$.
Consider bank $i \in [N]$ and assume that $\sum_{j \neq i} r_{ij}(X) > 0$. 
By the monotonicity and self-preferential property of the decomposition, we have
\[
\Lambda_i\of{X,r_{i\cdot}(X) + \sum_{j \neq i} r_{j\cdot}(X)} \leq \Lambda_i\of{X, \of{r_{ii}(X) + L \sum_{j \neq i} r_{ij}(X)}e_i + \sum_{j \neq i} r_{j\cdot}(X)}.
\]
Notably, $r_{ii}(X) + L \sum_{j \neq i} r_{ij}(X) < \sum_{k \in [N]} r_{ik}(X)$, contradicting the optimality of $r_{i\cdot}(X)$.
\end{proof}

\begin{remark}
Proposition~\ref{prop:self-preferential} can be extended to the $L = 1$ case in the following way:
The best response function $B_i(r_{-i};X)$ for bank $i$ always includes a point $r_{ij} = 0$ for $j \neq i$. As institutions would prefer to keep funds for themselves, all else being equal, we can consider banks to consider the selector $r_{ii} e_i \in B_i(r_{-i};X)$ thus reducing to the Nash allocation rules introduced in Definition~\ref{defn:nash-allocation}.
\end{remark}

\section{Constructing a Decomposition}\label{app:nash-decomposition}

In Section~\ref{app:general-minimal}, we will first address the decomposition selection problem for a generic aggregation function $\Delta\colon \R^N\times\R^N\to \R\cup\{-\infty\}$. Then, as in the setting of Remark~\ref{rem:select}, we will apply our results by assuming that there are aggregation functions $\lambda_1,\ldots,\lambda_N$ such that  $\Delta:=\Lambda-\sum_{i\in[N]}\lambda_i$ is an aggregation function. Finally, in Section~\ref{app:mf-decomp}, we will consider the mean-field aggregation function (see Example~\ref{ex:decomposition}\eqref{mf2}) and find a minimal decomposition for it within a class of symmetric decompositions.

\subsection{Existence of Minimal Decompositions}\label{app:general-minimal}

Let $\Delta\colon \R^N\times\R^N\to \R\cup\{-\infty\}$ be an aggregation function. Let us define the projection
\[
\mathcal{X}(\Delta)\coloneqq \{x\in\R^N\;|\;\exists m\in\R^N\colon (x,m)\in \dom\Delta \}=\{x\in \R^N\; | \; \dom\Delta|^{x}\neq \emptyset\}.
\]
Recall that, whenever $x\in\mathcal{X}(\Delta)$, the set $\dom\Delta|^x$ is a rectangular upper set. We will further assume this set has a (componentwise) minimum point as stated in the next assumption.

\begin{assumption}
There exists a function $\underline{r}\colon\mathcal{X}(\Delta)\to \R^N$ such that
\[
\dom\Delta|^x=\underline{r}(x)+\R^N_+
\]
for every $x\in \mathcal{X}(\Delta)$.
\end{assumption}
Under this assumption, it is clear that $\underline{r}$ is decreasing and lower semicontinuous, i.e., the epigraph $\{(x,m)\in \mathcal{X}(\Delta)\times\R^N \; | \; \underline{r}(x)\leq m\}=\dom\Delta$ is a closed set.

\begin{definition}\label{defn:ref-dec}
Let $\bar{r}\colon\mathcal{X}(\Delta)\to\R^N$ be a function such that $\underline{r}(x)\leq \bar{r}(x)$ for every $x\in\mathcal{X}(\Delta)$. Let $K\colon \mathcal{X}(\Delta)\to \R_+$ be a function. Let $\delta_1,\ldots,\delta_N\colon \R^N\times\R^N\to\R\cup\{-\infty\}$ be aggregation functions. The collection $\vec{\delta}=(\delta_i)_{i\in[N]}$ is called a \textbf{decomposition of $\Delta$ with reference pair $(\bar{r},K)$} if it satisfies the following properties:
\begin{enumerate}
\item \textbf{Decomposition:} $\sum_{i\in[N]}\delta_i(x,m)=\Delta(x,m)$ for every $(x,m)\in\R^N\times\R^N$.
\item \textbf{Consistent domain:} $\dom\delta_i=\dom\Delta$ for each $i\in[N]$.
\item \textbf{Self-feasible with $\bar{r}$:} For every $x\in \mathcal{X}(\Delta)$ and $i\in[N]$, it holds $\delta_i(x,(\bar{r}_i(x),\underline{r}_{-i}(x)))\geq 0$.
\item \textbf{Bounded by $K$ at $\underline{r}$:} For every $x\in\mathcal{X}(\Delta)$, it holds $\sum_{i\in[N]}\abs{\delta_i(x,\underline{r}(x))}\leq K(x)$.
\end{enumerate}
We denote by $\mathscr{D}(\bar{r},K)$ the set of all decompositions of $\Delta$ with reference pair $(\bar{r},K)$.
\end{definition}

Let $\mathscr{C}(\dom\Delta)$ be the space of all continuous functions $f\colon\dom\Delta\to\R$, or interchangeably, all functions $f\colon\R^N\times\R^N\to\R\cup\{-\infty\}$ with $\dom f=\dom\Delta$ such that the restriction of $f$ on $\dom f$ is continuous. We will equip $\mathscr{C}(\dom\Delta)$ with the topology of uniform convergence over compact subsets of $\dom\Delta$. One compatible metric for this topology can be defined by
\[
d(f_1,f_2)=\sum_{k=1}^\infty \frac{1}{2^k}\frac{\|(f_1-f_2)1_{\mathbb{B}_\infty(k)}\|_\infty}{1+\|(f_1-f_2)1_{\mathbb{B}_\infty(k)}\|_\infty}
\]
for every $f_1,f_2\in\mathscr{C}(\dom\Delta)$, where $\mathbb{B}_\infty(k)$ is the $\ell^\infty$ ball in $\R^N\times\R^N$ centered at the origin with radius $k\in \N$ and 
\[
\|f\|_\infty:=\sup_{(x,m)\in\dom \Delta}|f(x,m)|
\]
for every $f\in\mathscr{C}(\dom\Delta)$. Note that $(\mathscr{C}(\dom\Delta),d)$ is a complete separable metric space. We also define the product space $\mathscr{C}^N(\dom\Delta):=(\mathscr{C}(\dom\Delta))^N$ and equip it with the corresponding product topology.

From now on, we fix a reference pair $(\bar{r},K)$ as in Definition~\ref{defn:ref-dec}. From the definition of an aggregation function, it is clear that $\mathscr{D}(\bar{r},K)$ is a subset of $\mathscr{C}^N(\dom\Delta)$. We explore some useful properties of this set in the next few results.

\begin{lemma}\label{lem:D-closed}
$\mathscr{D}(\bar{r},K)$ is a closed subset of $\mathscr{C}^N(\dom\Delta)$.
\end{lemma}

\begin{proof}
Let $(\vec{\delta}^n)_{n\in\N}$ be a sequence in $\mathscr{D}(\bar{r},K)$ that converges to some $\vec{\delta}\in \mathscr{C}^N(\dom\Delta)$. In particular, for each $i\in[N]$, the sequence $(\delta_i^n)_{n\in\N}$ converges to $\delta_i$ pointwise on $\dom\Delta$. As pointwise convergence preserves monotonicity and concavity, the functions $\delta_1,\ldots,\delta_N$ are increasing and concave; hence, they are aggregation functions with common domain $\dom\Delta$. Pointwise convergence also implies that $\sum_{i\in[N]}\delta_i(x,m)=\Delta(x,m)$ for every $(x,m)\in\dom\Delta$.

To check self-feasibility with $\bar{r}$, let $x\in\mathcal{X}(\delta)$ and $i\in[N]$. Then, we have
\[
\delta_i(x,(\bar{r}_i(x),\underline{r}_i(x)))=\lim_{n\rightarrow\infty}\delta_i^n(x,(\bar{r}_i(x),\underline{r}_i(x)))\geq 0.
\]
Similarly, for every $x\in\mathcal{X}(\Delta)$, we have
\[
\sum_{i\in[N]}|\delta_i(x,\underline{r}(x))|=\lim_{n\rightarrow\infty}\sum_{i\in[N]}|\delta_i^n(x,\underline{r}(x))|\leq K(x).
\]
Hence, $\vec{\delta}\in\mathscr{D}(\overline{r},K)$, proving the closedness of $\mathscr{D}(\bar{r},K)$.
\end{proof}

\begin{lemma}\label{lem:D-bounded}
$\mathscr{D}(\bar{r},K)$ is pointwise bounded, i.e., for every $(x,m)\in\dom\Delta$, it holds
\[
\sup_{\vec{\delta}\in\mathscr{D}(\bar{r},K)}\max_{i\in[N]}|\delta_i(x,m)|<+\infty.
\]
\end{lemma}

\begin{proof}
Let $\vec{\delta}(\bar{r},K)$ and $(x,m)\in\dom\Delta$. In particular, $\underline{r}(x)\leq m$ since $m\in \dom\Delta|^x$. Then, by monotonicity, we have $\delta_i(x,\underline{r}(x))\leq \delta_i(x,m)$ for every $i\in[N]$ and
\[
\sum_{i\in[N]}\delta_i(x,\underline{r}(x))=\Delta(x,\underline{r}(x)).
\]
Hence, we have
\[
|\delta_i(x,m)|\leq (\delta_i(x,m)-\delta_i(x,\underline{r}(x)))+|\delta_i(x,\underline{r}(x))|\leq (\Delta(x,m)-\Delta(x,\underline{r}(x)))+\sum_{j\in[N]}|\delta_j(x,\underline{r}(x))|.
\]
By boundedness by $K$ at $\underline{r}$, this implies that
\[
\max_{i\in[N]}|\delta_i(x,m)|\leq (\Delta(x,m)-\Delta(x,\underline{r}(x)))+K(x)<+\infty.
\]
As the upper bound does not depend on the choice of $\vec{\delta}$, the claimed boundedness of $\mathscr{D}(\underline{r},K)$ at $(x,m)$ follows. 
\end{proof}

\begin{lemma}\label{lem:D-equi}
Let $i\in[N]$ and $k\in\N$ such that $\dom\Delta\cap\mathbb{B}_\infty(k)\neq\emptyset$. Then, the set $\{(\delta_i)|_{\mathbb{B}_\infty(k)}\;|\; 
\vec{\delta}\in\mathscr{D}(\bar{r},K)\}$ of restrictions on $\mathbb{B}_\infty(k)$ is uniformly equicontinuous as a subset of $\mathscr{C}(\dom\Delta\cap\mathbb{B}_\infty(k))$.
\end{lemma}

\begin{proof}
Let $\varepsilon>0$. Since $\Delta$ is a continuous function on $\dom\Delta$, its restriction $\Delta|_{\mathbb{B}_\infty(k)}$ on the compact set $\dom\Delta\cap \mathbb{B}_{\infty}(k)$ is uniformly continuous. Hence, there exists $\alpha>0$ such that the implication
\begin{equation}\label{eq:equi}
|x^1-x^2|\vee |m^1-m^2|\leq \alpha\quad\Rightarrow\quad |\Delta(x^1,m^1)-\Delta(x^2,m^2)|\leq\frac{\varepsilon}{2}
\end{equation}
holds for every $(x^1,m^1),(x^2,m^2)\in\dom\Delta\cap \mathbb{B}_{\infty}(k)$.

Let us fix $(x^1,m^1),(x^2,m^2)\in\dom\Delta\cap \mathbb{B}_{\infty}(k)$ such that $|x^1-x^2|\vee |m^1-m^2|\leq \alpha$. We define
\[
\tilde{x}:=x^1\vee x^2,\quad \tilde{m}:=m^1\vee m^2.
\]
By the monotonicity of $\dom\Delta$ and the definition of $\mathbb{B}_\infty(k)$, we have $(\tilde{x},\tilde{m})\in \dom\Delta\cap \mathbb{B}_\infty(k)$. A simple geometric observation yields that
\[
|x^1-\tilde{x}|\leq |x^1-x^2|\leq \alpha.
\]
Similarly, we have $|x^2-\tilde{x}|\leq \alpha$, $|m^1-\tilde{m}|\leq \alpha$, and $|m^2-\tilde{m}|\leq \alpha$. Then, by \eqref{eq:equi}, we have
\[
|\Delta(x^1,m^1)-\Delta(\tilde{x},\tilde{m})|\leq \frac{\varepsilon}{2},\quad |\Delta(x^2,m^2)-\Delta(\tilde{x},\tilde{m})|\leq \frac{\varepsilon}{2}.
\]
Moreover, since $(x^1,m^1)\leq (\tilde{x},\tilde{m})$ and $(x^2,m^2)\leq (\tilde{x},\tilde{m})$, by monotonicity, we have
\[
\delta_j(x^1,m^1)\vee \delta_j(x^2,m^2)\leq \delta_j(\tilde{x},\tilde{m})
\]
for every $j\in\N$ as well as
\[
\Delta(x^1,m^1)\vee \Delta(x^2,m^2)\leq \Delta(\tilde{x},\tilde{m}).
\]
Therefore,
\begin{align*}
|\delta_i(x^1,m^1)-\delta_i(x^2,m^2)|&\leq |\delta_i(x^1,m^1)-\delta_i(\tilde{x},\tilde{m})|+|\delta_i(\tilde{x},\tilde{m})-\delta_i(x^2,m^2)|\\
&=(\delta_i(\tilde{x},\tilde{m})-\delta_i(x^1,m^1))+(\delta_i(\tilde{x},\tilde{m})-\delta_i(x^2,m^2))\\
&\leq (\Delta(\tilde{x},\tilde{m})-\Delta(x^1,m^1))+(\Delta(\tilde{x},\tilde{m})-\Delta(x^2,m^2))\\
&\leq \frac{\varepsilon}{2}+\frac{\varepsilon}{2}=\varepsilon,
\end{align*}
which completes the proof of uniform equicontinuity.
\end{proof}

\begin{proposition}\label{prop:D-compact}
$\mathscr{D}(\bar{r},K)$ is a compact subset of $\mathscr{C}^N(\dom\Delta)$.
\end{proposition}

\begin{proof}
We use the generalization of Arzel\`{a}-Ascoli theorem for uniform convergence on compact sets, see \citet[Theorem~X.2.5]{bourbaki}. Lemma~\ref{lem:D-bounded} and Lemma~\ref{lem:D-equi} justify the hypotheses of this theorem. Hence, applying the theorem yields that $\{\delta_i\;|\; \vec{\delta}\in\mathscr{D}(\bar{r},K)\}$ is a precompact subset of $\mathscr{C}(\dom\Delta)$ for each $i\in[N]$. Then, $\bigtimes_{i\in[N]} \{\delta_i\;|\; \vec{\delta}\in\mathscr{D}(\bar{r},K)\}$ is a precompact subset of $\mathscr{C}^N(\dom\Delta)$. Finally, since $\mathscr{D}(\bar{r},K)$ is a closed subset of $\mathscr{C}^N(\dom\Delta)$ by Lemma~\ref{lem:D-closed} and
\[
\mathscr{D}(\bar{r},K)\subseteq \bigtimes_{i\in[N]} \{\delta_i\;|\; \vec{\delta}\in\mathscr{D}(\bar{r},K)\},
\]
we conclude that $\mathscr{D}(\bar{r},K)$ is a compact subset of $\mathscr{C}^N(\dom\Delta)$.
\end{proof}

\begin{definition}\label{defn:L}
Let $\vec{\delta}\in \mathscr{D}(\bar{r},K)$. For each $x\in \mathcal{X}(\Delta)$, let
\[
L_x(\vec{\delta}):=\inf\cb{L\geq 0\; \left| \; 
\begin{array}{l} \delta_i(x,m+L\alpha e_i)\geq \delta_i(x,m+\alpha e_j)\text{ for every }i,j\in[N]\\ \text{ with }i\neq j,\ m\in \dom\Delta|^x, \ \alpha\geq 0\end{array}
\right.
}
\]
Then, the minimal self-preferential constant of $\vec{\delta}$ is defined as
\[
L(\vec{\delta}):=\sup_{x\in \mathcal{X}(\Delta)}L_x(\vec{\delta}).
\]
\end{definition}

\begin{remark}\label{rem:L-defn}
Given $\vec{\delta}\in \mathscr{D}(\bar{r},K)$, we have $L(\vec{\delta})\in [0,+\infty]$ in general. For each $x\in\mathcal{X}(\Delta)$, the set in the definition of $L_x(\vec{\delta})$ is an interval that is unbounded from above due to monotonicity. Moreover, if this set is nonempty, i.e., $L_x(\vec{\delta})<+\infty$, then the infimum in the definition of $L_x(\vec{\delta})$ is attained by the continuity of aggregation functions.
\end{remark}

\begin{lemma}\label{lem:L-lsc}
The function $L\colon \mathscr{D}(\bar{r},K)\to [0,+\infty]$ is lower semicontinuous.
\end{lemma}

\begin{proof}
Since the pointwise supremum of a family of lower semicontinuous functions is lower semicontinuous, it is enough to show that $L_x\colon \mathscr{D}(\bar{r},K)\to [0,+\infty]$ is lower semicontinuous for an arbitrarily fixed $x\in \mathcal{X}(\Delta)$.

To that end, let $(\vec{\delta}^n)_{n\in\mathbb{N}}$ be a sequence in $\mathscr{D}(\bar{r},K)$ that converges to some $\vec{\delta}\in \mathscr{C}^N(\dom\Delta)$. We show that 
\[
L_x(\vec{\delta})\leq \liminf_{n\rightarrow\infty}L_x(\vec{\delta}^n)=:L_x^\ast.
\]
Since the case $L_x^\ast=+\infty$ is trivial, let us assume that $L_x^\ast<+\infty$. Then, there exists a subsequence $(\vec{\delta}^{n_k})_{k\in\N}$ such that $(L_x(\vec{\delta}^{n_k}))_{k\in\N}$ converges to $L_x^\ast$. Hence, for every $i,j\in[N]$ with $i\neq j$, $m\geq \underline{r}(x)$, and $\alpha\geq 0$, we have
\begin{align*}
\delta_i(x,m+L_x^\ast)&=\lim_{k\rightarrow\infty}\delta_i^{n_k}(x,m+L_x(\vec{\delta}^{n_k})\alpha e_i)\\
&\geq \lim_{k\rightarrow\infty}\delta_i^{n_k}(x,m+\alpha e_i)\\
&=\delta_i(x,m+\alpha e_j).
\end{align*}
In this calculation, the first equality follows since the sequence $(L_x(\vec{\delta}^{n_k}))_{k\in\N}$ is bounded and $(\delta_i^{n_k})_{k\in\N}$ converges to $\delta_i$ uniformly on compact sets, the inequality is by the attainment of each $L_x(\vec{\delta}^{n_k})$ as an infimum (see Remark~\ref{rem:L-defn}), and the second equality follows since uniform convergence on compact sets implies pointwise convergence. Hence, we conclude that $L_x^\ast\geq L_x(\vec{\delta})$ and the lower semicontinuity of $L_x(\cdot)$ follows.
\end{proof}

\begin{theorem}\label{thm:exist-dec}
Assume that the function $L\colon \mathscr{D}(\bar{r},K)\to [0,+\infty]$ is proper, i.e., there exists $\vec{\delta}^0\in \mathscr{D}(\bar{r},K)$ such that $L(\vec{\delta}^0)<+\infty$. Then, the optimization problem
\begin{equation}\label{eq:MD}
\text{minimize}\ L(\vec{\delta})\ \text{subject to }\ \vec{\delta}\in \mathscr{D}(\bar{r},K)\tag{MD}
\end{equation}
has an optimal solution.
\end{theorem}

\begin{proof}
Since $\vec{\delta}^0$ is a feasible solution of \eqref{eq:MD}, the restricted problem
\[
\text{minimize}\ L(\vec{\delta})\ \text{subject to }\ L(\vec{\delta})\leq L(\vec{\delta}^0),\quad \vec{\delta}\in \mathscr{D}(\bar{r},K)
\]
has the same optimal value and the same set of optimal solutions as the original one. Moreover, the new feasible region $\{\vec{\delta}\in \mathscr{D}(\bar{r},K)\;|\; L(\vec{\delta})\leq L(\vec{\delta}^0)\}$ is closed since the function $L(\cdot)$ is lower semicontinuous by Lemma~\ref{lem:L-lsc}. Then, this set is also compact since $\mathscr{D}(\bar{r},K)$ is compact by Proposition~\ref{prop:D-compact}. The function $L$ is finitely valued on this set. Therefore, an optimal solution exists for both problems.
\end{proof}

We end this section by formulating an analogous framework for single-element aggregation functions. To that end, let $\bar\Delta\colon\R^N\to\R\cup\{-\infty\}$ be a single-element aggregation function such that
\[
\dom\bar\Delta=\underline{x}+\R^N_+
\]
for some $\underline{x}\in\R^N$.

\begin{definition}\label{defn:ref-dec-single}
Let $\bar{x}\in \dom\bar\Delta$ and $k>0$. Let $\bar{\delta}_1,\ldots,\bar{\delta}_N\colon \R^N\to\R\cup\{-\infty\}$ be single-element aggregation functions. The collection $\vec{\bar{\delta}}=(\bar\delta_i)_{i\in[N]}$ is called a \textbf{single-element decomposition of $\bar\Delta$ with reference pair $(\bar{x},k)$} if it satisfies the following properties:
\begin{enumerate}
\item \textbf{Decomposition:} $\sum_{i\in[N]}\bar\delta_i(x)=\bar\Delta(x)$ for every $x\in\R^N$.
\item \textbf{Consistent domain:} $\dom\bar\delta_i=\dom\bar\Delta$ for each $i\in[N]$.
\item \textbf{Self-feasible with $\bar{x}$:} For every $x\in\dom\bar\Delta$ and $i\in[N]$, it holds $\bar\delta_i(\bar{x}_i,\underline{x}_{-i})\geq 0$.
\item \textbf{Bounded by $k$ at $\underline{x}$:} For every $x\in\dom\bar\Delta$, it holds $\sum_{i\in[N]}\abs{\bar\delta_i(\underline{x})}\leq k$.
\end{enumerate}
We denote by $\bar{\mathscr{D}}(\bar{x},k)$ the set of all single-element decompositions of $\bar\Delta$ with reference pair $(\bar{x},k)$.

Given $\vec{\bar\delta}\in \bar{\mathscr{D}}(\bar{x},k)$, the minimal self-preferential constant of $\vec{\bar\delta}$ is defined as
\[
\bar{L}(\vec{\bar\delta}):=\inf\cb{L\geq 0\; \left| \; 
\begin{array}{l} \bar{\delta}_i(x+L\alpha e_i)\geq \bar{\delta}_i(x+\alpha e_j)\text{ for every }i,j\in[N]\\ \text{ with }i\neq j,\ x\in \dom\bar\Delta, \ \alpha\geq 0\end{array}.
\right.
}
\]
\end{definition}

\begin{theorem}\label{thm:exist-dec-single}
Assume that the function $\bar{L}\colon \bar{\mathscr{D}}(\bar{x},k)\to [0,+\infty]$ is proper, i.e., there exists $\vec{\bar\delta}^0\in \bar{\mathscr{D}}(\bar{x},k)$ such that $\bar{L}(\vec{\bar\delta}^0)<+\infty$. Then, the optimization problem
\begin{equation}\label{eq:MD-bar}
\text{minimize}\ \bar{L}(\vec{\bar\delta})\ \text{subject to }\ \vec{\bar{\delta}}\in \bar{\mathscr{D}}(\bar{x},k)\tag{$\overline{\text{MD}}$}
\end{equation}
has an optimal solution.
\end{theorem}

\begin{proof}
The proof follows by the simplified ``single-element version'' versions of the lemmata leading to the proof of Theorem~\ref{thm:exist-dec}. To avoid repetitions, we omit the details.
\end{proof}

\begin{remark}\label{rem:neuralnet}
Within this section we introduced Problem~\eqref{eq:MD} to define the minimal decomposition and present theoretical results related to the existence of a minimizer. However, as an infinite-dimensional problem, the computation of such a minimal decomposition is, generically, beyond the scope of this work. To approximate this minimal decomposition, we suggest utilizing input convex neural networks (see, e.g.,~\cite{amos2017input}) which minimizes the self-preferential constant (computed via the partial derivatives) at a random sampling of training points with penalization for violations of monotonicity and exact decomposition of the full aggregation function.
\end{remark}

\begin{remark}
We note that self-feasibility as defined in Assumption~\ref{asmp:decomposition} is with respect to a strict inequality while, herein, with a non-strict inequality. Notably, the full decomposition $(\Lambda_i := \lambda_i + \delta_i)_{i \in [N]}$ is self-feasible with the strict inequality so long as the natural decomposition $(\lambda_i)_{i \in [N]}$ is a non-trivial decomposition. In the case that $\lambda_i \equiv 0$ for every bank $i$, then existence of a Nash equilibrium can be proven so long as the minimal decomposition of $\Delta = \Lambda$ is self-preferential with $L < 1$ even with the non-strict inequality in self-feasibility; this follows via the same continuity argument as in the proof of Theorem~\ref{thm:exist-unique:2}.
\end{remark}

\subsection{Minimal Decomposition of the Mean-Field Aggregation Function}\label{app:mf-decomp}

Within this section, we want to explore the minimal decomposition (Section~\ref{app:general-minimal}) for the mean-field aggregation function (Example~\ref{ex:se-aggregation}\eqref{mf}), i.e., $\bar\Lambda^{MF}(x) = \sum_{i \in [N]} u_i(x_i) + \bar u(\frac{1}{N}\sum_{i \in [N]}x_i)$ for any $x \in \R^N$. In particular, we want to consider the decomposition as follows from Example~\ref{ex:decomposition}\eqref{mf2} with $\bar\lambda_i(x) := u_i(x_i)$. This approach produces a remainder term $\bar\Delta(x) := \bar u(\frac{1}{N} \sum_{i \in [N]} x_i)$. In this section, we provide the minimal single-element decomposition for the sensitive version of this aggregation function, i.e., $\Delta(x,m) := \bar\Delta(x+m)$, following Problem~\eqref{eq:MD}.

\begin{proposition}\label{prop:mf-decomp}
Consider the single-element aggregation function $\bar\Delta(x) = \bar u(\frac{1}{N} \sum_{i \in [N]} x_i)$ with $\dom\bar\Delta = \R^N_+$, where $\bar u$ is a strictly concave, increasing, and twice continuously differentiable function such that $\lim_{z \to \infty} u'(z) = 0$. Let $\vec{\bar{\delta}}=(\bar\delta_i)_{i \in [N]}$ be a single-element decomposition of $\bar\Delta$ satisfying:
\begin{enumerate}
\item \textbf{Symmetry:} For any distinct banks $i,j,k \in [N]$, $\bar\delta_i(x) = \bar\delta_i(x_{j \leftrightarrow k})$, where $x_{j \leftrightarrow k} = x + (x_k - x_j)e_j + (x_j - x_k)e_k$;
\item \textbf{Twice differentiable:} $\bar\delta_i$ is twice continuously differentiable for each $i\in[N]$.
\end{enumerate}
The sensitive version of this decomposition, $\delta_i(x,m) := \bar\delta_i(x+m)$ for any $i$, has a self-preferential constant $L \geq 1$.
\end{proposition}

\begin{proof}
Let us fix some single-element decomposition $(\bar\delta_i)_{i \in [N]}$ of $\bar\Delta$. Recall that, by definition, $\sum_{i \in [N]} \bar\delta_i = \bar\Delta$ (\emph{decomposition}) and each $\bar\delta_i$ is increasing in all components (\emph{monotonicity}).

To complete this proof, we will first prove that, for any $i$, $\bar\delta_i(x) = \hat\delta_i(\sum_{k \in [N]} x_k) + \beta x_i$ for some concave function $\hat\delta_i: \R_+ \to \R$ and constant $\beta \in \R$.
Second, we will prove that $\beta \leq 0$. Finally, with this structural form, we will demonstrate that the self-preferential constant $L$ can\emph{not} be less than $1$, which completes the proof.

\begin{enumerate}
\item \textbf{Structural Form:} Fix banks $j \neq k$ and define $x(t) := \bar x + td$ for some $\bar x \in \R^N_+$ and $\sum_{i \in [N]}d_i = 0$. Note that $\sum_{i \in [N]} x_i(t) = \sum_{i \in [N]} \bar x_i$ for any feasible $t$. Let $\phi_i(t) := \bar\delta_i(x(t))$. By construction, $\phi_i$ is concave for every bank $i$, i.e., $\phi_i''(t) \leq 0$. Furthermore, by the decomposition property, $\sum_{i \in [N]} \phi_i(t) = \bar u(\frac{1}{N} \sum_{i \in [N]} \bar x_i)$ for any feasible $t$. Therefore, it must follow that $\sum_{i \in [N]} \phi_i''(t) = 0$. As the sum of non-positive elements that sum to 0, it must therefore follow that $\phi_i''(t) = 0$ for every bank $i$ and every feasible $t$.

Following Proposition~\ref{prop:structural-form} below, any function with zero second derivatives along all directions $d$ such that $\sum_{k \in [N]}d_k = 0$ must have the form 
\[
\bar\delta_i(x) = G_i\of{\sum_{k \in [N]} x_k} + \sum_{k \in [N]} \bar\beta_{ik} x_k
\]
for any bank $i$ and $x \in \R^N_+$. In addition, by symmetry, it must follow that $\bar\beta_{ij} = \bar\beta_{ik} =: \hat\beta_i$ for any $j,k \neq i$. Therefore,
    \[
    \bar\delta_i(x) = G_i\of{\sum_{k \in [N]} x_k} + \bar\beta_{ii}x_i + \hat\beta_i\sum_{j \neq i} x_j 
    = \underbrace{G_i\of{\sum_{k \in [N]} x_k} + \hat\beta_i\sum_{k \in [N]} x_k}_{=: h_i(\sum_{k \in [N]} x_k)} + \underbrace{(\bar\beta_{ii} - \hat\beta_i)}_{=: \beta_i} x_i.
    \]
It remains to prove that $\beta_i = \beta_j$ for any $i,j \in [N]$. By the decomposition property, $\bar u(\frac{1}{N} \sum_{k \in [N]} x_k) = \sum_{i \in [N]} h_i(\sum_{k \in [N]} x_k) + \sum_{i \in [N]} \beta_i x_i$. However, since the left-hand side only depends on the sum $S = \sum_{k \in [N]} x_k$ it must follow that the same is true of the right-hand side, i.e., $\beta_i = \beta_j$ for any $i,j \in [N]$.

\item \textbf{Monotonicity:} Consider some fixed vector $x \in \R^N_+$ and define $S := \sum_{k \in [N]} x_k$. Note that, by the construction of the structural form for the decomposition and efficiency, $\bar u(\frac{1}{N} S) = \sum_{i \in [N]} h_i(S) + \beta S$. By taking derivatives, this also implies that $\frac{1}{N} \bar u'(\frac{1}{N} S) = \sum_{i \in [N]} h_i'(S) + \beta$. Furthermore, by monotonicity of $\bar\delta_i$, it must follow that $h_i: \R_+ \to \R$ is also non-decreasing, i.e., $h_i'(S) \geq 0$. Therefore, by the asymptotic behavior of $\bar u'$, it must follow that $\beta \leq 0$.

\item \textbf{Self-Preferential Constant:} Finally, note that
\[
\frac{\partial \bar\delta_i}{\partial x_i}(x) = h_i'\of{\sum_{k \in [N]} x_k} + \beta \leq h_i'\of{\sum_{k \in [N]} x_k} = \frac{\partial \bar\delta_i}{\partial x_j}(x)
\]
for any $i \neq j$ by $\beta \leq 0$. Therefore, it must follow that the self-preferential constant $L$ must exceed $1$.
\end{enumerate}
\end{proof}

\begin{example}\label{ex:mf-decomp}
Proposition~\ref{prop:mf-decomp} states that no symmetric decomposition can have a self-preferential constant smaller than $1$. Notably, any decomposition $\bar\delta_i^\alpha(x) := \alpha_i \bar u(\frac{1}{N}\sum_{k \in [N]} x_k)$ for $\alpha$ in the unit simplex attains this minimal possible self-preferential constant of $1$. For stronger symmetry, we recommend taking $\alpha := \vec{1}/N$ as done in Example~\ref{ex:nash}\eqref{nash-mf}.
\end{example}

\begin{proposition}\label{prop:structural-form}
Let $f: \R^N \to \R \cup \{-\infty\}$ be a concave, twice continuously differentiable function such that $d^\T\nabla^2 f(x)d = 0$ for any $x \in \operatorname{int}\dom f$ and $d \in D := \{d \in \R^N \; | \; \sum_{i \in [N]}d_i = 0\}$. Then $f(x) = G(\sum_{i \in [N]} x_i) + \beta^\T x$ for some function $G: \R \to \R \cup \{-\infty\}$ and vector $\beta \in \R^N$ with $\sum_{i \in [N]} \beta_i = 0$.
\end{proposition}

\begin{proof}
In order to prove this result, we will first prove that $\nabla^2 f(x) = \bar c(x)\vec{1}\vec{1}^\T$, $x\in \operatorname{int}\dom f$, for some univariate mapping $\bar c: \R^N \to \R \cup \{-\infty\}$. With this Hessian, we are able to deduce a structural form for the gradient $\nabla f(x)$. Finally, we use these results to prove the proposed structural form for $f$. For the purposes of this proof, fix $x \in \operatorname{int}\dom f$.
\begin{enumerate}
\item \textbf{Hessian:} Fix some vector $v \in \R^N$ and note that it can be decomposed as $v = \alpha\vec{1} + d$ for $\alpha := \vec{1}^\T v/N \in \R$ and $d := v - \alpha\vec{1} \in D$. To provide a structure to the Hessian $\nabla^2 f(x)$, consider $v^\T\nabla^2 f(x) v$:
\begin{align*}
v^\T \nabla^2 f(x) v &= (d^\T \nabla^2 f(x) d) + 2\alpha(d^\T \nabla^2 f(x) \vec{1}) + \alpha^2(\vec{1}^\T\nabla^2 f(x) \vec{1})\\
&= 2\alpha (v^\T \nabla^2 f(x) \vec{1}) - \alpha^2 (\vec{1} \nabla^2 f(x) \vec{1}) \\
&= \frac{2}{N}(\vec{1}^\T v)(v^\T \nabla^2 f(x) \vec{1}) - \frac{1}{N^2}(\vec{1}^\T v)^2 (\vec{1} \nabla^2 f(x) \vec{1}) \\
&= \sum_{i \in [N]} \sum_{j \in [N]} \left[\frac{1}{N}\left((\nabla^2 f(x) \vec{1})_i + (\nabla^2 f(x) \vec{1})_j\right) - \frac{1}{N^2}(\vec{1}^\T \nabla^2 f(x) \vec{1})\right]v_i v_j.
\end{align*}
As this holds for any vector $v \in \R^N$, it must follow that $\nabla^2 f(x) = a(x)\vec{1}^\T + \vec{1}a(x)^\T + c(x)\vec{1}\vec{1}^\T$ for vector-valued $a(x) := \frac{1}{N}(\nabla^2 f(x)\vec{1})$ and scalar-valued $c(x) := -\frac{1}{N^2}(\vec{1}^\T \nabla^2 f(x) \vec{1})$. Taking advantage of this structural form, recall that $f$ is concave. Therefore,
\begin{align*}
0 &\geq v^\T \nabla^2 f(x) v = (\alpha\vec{1} + d)^\T [a(x)\vec{1}^\T + \vec{1}a(x)^\T + c(x)\vec{1}\vec{1}^\T](\alpha\vec{1} + d) \\
&= N\left[\alpha^2(Nc(x) + 2\vec{1}^\T a(x)) + 2\alpha(d^\T a(x))\right].
\end{align*}
Noting that $Nc(x) + 2\vec{1}^\T a(x) \leq 0$ by considering the same quadratic in direction $v_\alpha := \alpha\vec{1}$, it must follow that $d^\T a(x) = 0$ for any $d \in D$ (or else there exists some $d \in D$ such that the quadratic $v^\T\nabla^2 f(x)v > 0$ for $\alpha$ sufficiently small). As a consequence, $a(x) = \bar a(x) \vec{1}$ for some scalar function $\bar a: \R^N \to \R$. In particular, this implies that $\nabla^2 f(x) = \bar c(x) \vec{1}\vec{1}^\T$ for $\bar c(x) := c(x) + 2\bar a(x)$.
\item \textbf{Gradient:} As a direct consequence of the structure of the Hessian, for any $d,d' \in D$:
\begin{align*}
\frac{\partial}{\partial t}d^\T\nabla f(x + td') &= d^\T\nabla^2 f(x + td') d' = d^\T [\bar c(x+td')\vec{1}\vec{1}^\T] d' = 0.
\end{align*}
Therefore, $d^\T\nabla f(x) = d^\T\nabla f(x + d')$ for any $d,d' \in D$. In particular, this implies that $d^\T\nabla f(x) = d^\T\bar\beta(\sum_{i \in [N]} x_i)$ for some mapping $\bar\beta: \R \to \R^N$. In fact, we can take $\bar\beta$ to be orthogonal to $\vec{1}$ without loss of generality, i.e., $\vec{1}^\T\bar\beta(S) = 0$ for any feasible $S$. In this way, we know that $\nabla f(x) = \gamma(x)\vec{1} + \bar\beta(\sum_{i \in [N]} x_i)$.
\item \textbf{Structural Form:} Fix some $\alpha \in \R$ and $d \in D$. Recall that $d^\T \nabla^2 f(x) d = 0$, therefore we can represent $f(\alpha\vec{1}+d)$ by its first-order Taylor expansion, i.e.,
\begin{align*}
f(\alpha\vec{1} + d) &= f(\alpha\vec{1}) + d^\T\nabla f(\alpha\vec{1}) = f(\alpha\vec{1}) + d^\T\bar\beta(N\alpha).
\end{align*}
Let $\alpha = \frac{1}{N}\sum_{i \in [N]}$ and $d = x - \alpha\vec{1}$ then the above Taylor expansion can be reformulated as:
\begin{align*}
f(x) &= \underbrace{f\of{\frac{1}{N}\sum_{i \in [N]}x_i \vec{1}}}_{=: \, G(\sum_{i \in [N]} x_i)} + x^\T \bar\beta(\sum_{i = 1}^N x_i) - \frac{1}{N} \sum_{i \in [N]} x_i \underbrace{\vec{1}^\T \bar\beta\of{\sum_{i \in [N]} x_i}}_{= \, 0}.
\end{align*}
That is, $f(x) = G(\sum_{i \in [N]} x_i) + \bar\beta(\sum_{i \in [N]} x_i)^\T x$. It remains to show that $\bar\beta(S) = \beta \in \R^N$ for any feasible $S \in \R$. 
Let $S(x) := \sum_{i \in [N]} x_i$ for simplicity of notation. Using the deduced form $f(x) = G(S(x)) + \bar\beta(S(x))^\T x$ and the Hessian we deduced in part 1 of this proof, we find:
\[\bar c(x) = G''(S(x)) + \bar\beta_i'(S(x)) + \bar\beta_j'(S(x)) + \bar\beta''(S(x))^\T x\]
for any $x \in \operatorname{int}\dom f$. Since $\bar c$ does not depend on the indices, it must follow that $\bar\beta_i'(S) + \bar\beta_j'(S) = \bar\beta_i'(S) + \bar\beta_k'(S)$ for any $i,j,k \in [N]$, i.e., $\bar\beta_j'(S) = \bar\beta_k'(S)$ for every pair of indices $j,k \in [N]$. Notably, this can only occur if $\bar\beta'(S) = \beta'(S)\vec{1}$ for some scalar function $\beta': \R \to \R \cup \{-\infty\}$. Furthermore, because $\sum_{i \in [N]} \bar\beta_i(\cdot) = 0$, we know that $\sum_{i \in [N]} \bar\beta_i'(S) = N\beta'(S) = 0$, i.e., $\beta'(S) = 0$ for any $S$. Integrating, it must follow that $\bar\beta(S) = \beta \in \R^N$ is constant and the proof is completed.
\end{enumerate}
\end{proof}

\section{Calibration to European Banking Dataset}\label{sec:calibration}
Within Section~\ref{sec:EN-numeric-EBA}, we reconstruct the interbank liability matrix following the approach outlined in \cite{gandy2017bayesian}. In doing so, we exactly copy the hyperparameters utilized in \cite{feinstein2019obligations}.

First, we note that the 2011 EBA stress testing dataset provides the total assets $T_i$, capital $c_i$, and interbank liabilities $\sum_{j \in [N]} \bar p_{ij}$ for all banks $i\in[N]$ in the dataset.\footnote{Due to technical issues with the calibration, we only consider $N=87$ of the $90$ institutions in the dataset. DE029, LU45, and SI058 are not included within this calibration.}
For simplicity, we set the total interbank assets to be equal to the total interbank liabilities, i.e., $\sum_{j \in [N]} \bar p_{ij} = \sum_{j \in [N]} \bar p_{ji}$ for every bank $i$ with slight perturbations as discussed within \cite{gandy2017bayesian}. For simplicity, we assume that the difference between the total assets ($T_i$) and the interbank assets ($\sum_{j \in [N]} \bar p_{ji}$) are the (initial) external assets ($x_i := T_i - \sum_{j \in [N]} \bar p_{ji}$). Similarly, all liabilities that are not owed to other banks are external ($\bar p_{i0} := T_i - \sum_{j \in [N]} \bar p_{ij} - c_i$).

With these parameters, we reconstruct the full nominal liabilities matrix, i.e., the exact values of $\bar p_{ij}$ for $i \neq j$, using the methodology of \cite{gandy2017bayesian}. Though a stochastic framework, within Section~\ref{sec:EN-numeric-EBA}, we only consider a single realization of this liabilities matrix. For this reconstruction we consider hyperparameters $p = 0.5$, $\text{thinning} = 10^4$, $n_{\text{burn-in}} = 10^9$, and $\lambda = \frac{pn(n-1)}{\sum_{i \in [N]} \sum_{j \in [N]} \bar p_{ij}} \approx 0.00122$.

\end{document}